\documentclass[conference]{IEEEtran}
\IEEEoverridecommandlockouts 

\usepackage{amsmath, amsthm, amssymb}
\usepackage{dsfont}
\usepackage{xspace}

\csname@@input\endcsname xypic \xyoption{arrow} \xyoption{tips}
\SelectTips{cm}{10} \UseTips
\newdir{ >}{{}*!/-5pt/\dir{>}}

\DeclareFontFamily{U}{cmr}{} \DeclareFontShape{U}{cmr}{a}{n}{<5>
cmb5<7>cmb7<10>cmb10}{}

\DeclareFontFamily{U}{msa}{} \DeclareFontShape{U}{msa}{m}{n}{<-6
>msam5<6-8.5>msam7<8.5->msam10}{} \DeclareSymbolFont{latex-font
msa}{U}{msa}{m}{n} \DeclareFontFamily{U}{msb}{}
\DeclareFontShape{U}{msb}{m}{n}{<-6>msbm5<6-8.5>msbm7<8.5->msbm1
0}{} \DeclareSymbolFont{latex-font msb}{U}{msb}{m}{n}

\DeclareMathAlphabet{\mathscr}{U}{rsfs}{m}{n}
\DeclareMathAlphabet{\mathbt}{OT1}{cmr}{a}{n}

\newtheorem{X}{}[section]
\newtheorem{XCorollary}[X]{Corollary}
\newtheorem{XProposition}[X]{Proposition}
\newtheorem{XTheorem}[X]{Theorem} 
\newtheorem{XLemma}[X]{Lemma}

\theoremstyle{definition}

 \newtheorem{XNotation}[X]{Notation} \newtheorem{XDefinition}[X]{Definition}
\newtheorem{XExamples}[X]{Examples} \newtheorem
{XExample}[X]{Example} \newtheorem{XRemark}[X]{Remark}
\newtheorem{XAssumptions}[X]{Assumptions}

\newtheorem{XConstruction}[X]{Construction}

\DeclareMathSymbol\emptyset\mathord{latex-font msb}{"3F}
\DeclareMathSymbol\onto\mathrel{latex-font msa}{"10}
\DeclareMathSymbol\square\mathrel{latex-font msa}{"03}
\DeclareMathAccent{\bigwidehat}{\mathord}{latex-font msb}{"5B}

\def\hyph{-\penalty0\hskip0pt\relax}
\def\cpo{cpo\xspace}
\def\refeq#1{(\ref{#1})}
\def\epsilon{\varepsilon}

\renewcommand{\rho}{\varrho}
\def\ol{\overline}
\def\wh{\widehat}

\def\out{\mathsf{out}}
\def\init{\mathsf{in}}

\newcommand{\takeout}[1]{\empty}

\newcommand{\ra}{\rightarrow}
\newcommand{\xra}{\xrightarrow}

\newcommand{\seq}{\subseteq}

\newcommand{\Alg}{\mathbf{Alg}\,}
\newcommand{\Set}{\mathbf{Set}}

\newcommand{\BA}{\mathbf{BA}}
\newcommand{\BR}{\mathbf{BR}}
\DeclareMathOperator{\JSLtc}{\mathbf{JSL_{01}}}
\DeclareMathOperator{\JSLnc}{\mathbf{JSL}}

\newcommand{\DL}{\mathbf{DL}_{01}}

\newcommand{\Poset}{\mathbf{Pos}}

\newcommand{\ACat}{\mathscr{A}}
\newcommand{\BCat}{\mathscr{B}}
\newcommand{\Cat}{\mathscr{C}}

\newcommand{\DCat}{\mathscr{D}}

\newcommand{\id}{\mathsf{id}}

\newcommand{\Mon}[1]{#1\mathchar`-\mathbf{Mon}}
\newcommand{\Coalg}[1]{\mathbf{Coalg}\,#1}

\newcommand{\FCoalg}[1]{\mathbf{Coalg}_f\,#1}

\newcommand{\FAlg}[1]{\mathbf{Alg}_{f}\,#1}

\DeclareMathOperator{\Kl}{\mathbf{Kl}}

\newcommand{\Vect}[1]{\mathbf{#1\mathchar`-Vec}}

\newcommand{\JSL}{{\mathbf{JSL}_0}}

\newcommand{\PSet}{\mathbf{Set}_\star}

\newcommand{\under}[1]{|#1|}

\newcommand{\hookto}{\hookrightarrow}
\newcommand{\epito}{\twoheadrightarrow}

\newcommand{\monoto}{\rightarrowtail}

\newcommand{\V}{\mathfrak{V}}

\newcommand{\rev}{\mathsf{rev}}

\newcommand{\Pow}{\mathcal{P}}

\newcommand{\Int}{\mathds{Z}_2}

\newcommand{\one}{\mathbf{1}}

\newcommand{\two}{\mathbf{2}}

\usepackage{url}

\usepackage{tabularx}
\usepackage{array,booktabs}
\usepackage{microtype}

\title{Varieties of Languages in a Category}

\author{\IEEEauthorblockN{Ji\v{r}\'\i~Ad\'amek, Robert S.~R. Myers, Henning Urbat}
\IEEEauthorblockA{Institut f\"ur Theoretische Informatik \\ 
  Technische Universit\"at Braunschweig, Germany}
\IEEEauthorblockN{Stefan Milius\thanks{Stefan Milius acknowledges support by the Deutsche Forschungsgemeinschaft (DFG) under project MI 717/5-1}}
\IEEEauthorblockA{Lehrstuhl f\"ur Theoretische Informatik \\
Friedrich-Alexander Universit\"at Erlangen-N\"urnberg, Germany}
}
\begin{document}

\maketitle

%
%
\pagestyle{plain}
\thispagestyle{plain}

\begin{abstract} Eilenberg's variety theorem, a centerpiece of algebraic automata theory, establishes a bijective correspondence between varieties of languages and pseudovarieties of monoids.
In the present paper this result is generalized to an abstract pair of algebraic categories: we introduce varieties of languages in a category~${\mathscr C}$, and prove
that they correspond to pseudovarieties of monoids in a closed monoidal category $\DCat$, provided that $\Cat$ and $\DCat$ are dual on the level of finite objects. By suitable choices of these categories our result uniformly covers Eilenberg's theorem and three variants due to Pin, Pol\'{a}k and Reutenauer, respectively, and yields new Eilenberg-type correspondences.
\end{abstract}


\section{Introduction}

Algebraic automata theory investigates the relation between regular languages and algebraic structures like  monoids, semigroups, or semirings. A major result
concerns \emph{varieties of languages}. These are classes of regular languages
closed under 
\begin{itemize} \item[(a)] boolean operations (union, intersection and complement), \item[(b)] derivatives,
i.e., with every language $L\seq \Sigma^*$ a variety
contains its \emph{left derivatives} $a^{-1}L=\{w\in\Sigma^*:aw\in L\}$ and
\emph{right derivatives} $La^{-1}=\{w\in\Sigma^*:wa\in L\}$ for
all  $a\in \Sigma$, and \item[(c)] preimages under monoid morphisms
$f:{\Delta}^*\to{\Sigma}^*$. \end{itemize}
Eilenberg proved in his monograph \cite{E} that the lattice of all varieties of languages is isomorphic to
the lattice of all \emph{pseudovarieties of monoids}, these being classes of finite monoids closed under finite products, submonoids and homomorphic images. Several variants of Eilenberg's theorem are known in the literature, altering the closure properties in the definition of a variety and replacing monoids by other algebraic structures. Pin~\cite{Pin95} introduced \emph{positive varieties of languages} where in (a) the closure under complement is omitted, and he proved a bijective correspondence to
pseudovarieties of \emph{ordered} monoids. Later Pol\'ak  \cite{P} further weakened (a) by also omitting closure under intersection, and the
 resulting \emph{disjunctive varieties of languages} correspond to pseudovarieties of idempotent semirings. Reutenauer \cite{R} studied a concept of variety where (a) is replaced  by closure under symmetric difference, and obtained a correspondence to pseudovarieties of algebras over the binary field $\Int$. (In fact Reutenauer considered algebras over arbitrary fields $\mathbf{K}$ and varieties of formal power series in lieu of languages). Finally, a new Eilenberg-type theorem we derive below deals with varieties of languages defined by closure under intersection and symmetric difference in lieu of (a), and relates them to pseudovarieties of monoids with $0$.

In this paper a categorical result is presented that covers Eilenberg's theorem and all its
variants uniformly, and exhibits
new applications. Our overall approach to algebraic automata theory may be subsumed by the ``equation''
\[ \textbf{automata theory} ~\boldsymbol{=}~\textbf{duality} ~\boldsymbol{+}~ \textbf{monoidal structure.} \]
The idea is to take a category $\Cat$ (where automata and languages live) and a closed monoidal category $\DCat$ (where monoids live) with the property that $\Cat$ and $\DCat$  are \emph{predual}. Specifically, in our setting $\Cat$ and $\DCat$ will be locally finite varieties of algebras or ordered algebras (i.e., all finitely generated algebras are finite), and preduality means that the full subcategories of finite algebras are dually equivalent. Moreover, the monoidal structure of $\DCat$ is given by the usual tensor product of algebras. 

All the Eilenberg-type correspondences mentioned above fit into this categorical framework. For example, the categories $\Cat$ of boolean algebras and $\DCat$ of sets are predual via Stone duality, and $\DCat$-monoids are ordinary monoids: this is the setting of Eilenberg's original result. The category $\Cat$ of distributive
lattices with $0$ and~$1$ is predual to the category $\DCat$ of posets via Birkhoff duality \cite{B2}, and $\DCat$-monoids are ordered monoids, which leads to Pin's result for pseudovarieties of ordered monoids
\cite{Pin95}. The category $\Cat$ of join-semilattices with $0$ is self-predual  (i.e., one takes $\DCat=\Cat$), and ${\mathscr D}$-monoids are precisely idempotent
semirings. This is the framework for Pol\'ak  \cite{P}.
For Reutenauer's result \cite{R} one takes the category $\Cat$ of
vector spaces over a finite field $\mathbf{K}$ which is also self-predual (i.e., $\DCat=\Cat$), and observes that ${\mathscr
D}$-monoids are precisely $\mathbf{K}$-algebras. Lastly, our new example concerning pseudovarieties of  monoids with $0$ takes as ${\mathscr C}$ non-unital boolean rings and as ${\mathscr D}$ pointed sets.

Apart from preduality, the heart of the matter is a coalgebraic characterization of the closure properties defining varieties of languages. We model deterministic $\Sigma$-automata in a locally finite variety $\Cat$ as coalgebras $Q\ra T_\Sigma Q$ for the endofunctor
\[ T_\Sigma: \Cat\ra\Cat,\quad T_\Sigma Q = O_\Cat \times Q^\Sigma,\]
where $O_\Cat$ is a fixed two-element algebra in $\Cat$ representing final and non-final states. In particular, the set of all regular languages over $\Sigma$ carries the structure of a $T_\Sigma$-coalgebra whose transitions are given by left derivatives $L\to a^{-1}L$ for $a\in\Sigma$, and whose final states are the languages containing the empty word. This coalgebra admits an abstract characterization as the \emph{rational fixpoint} $\rho T_\Sigma$ of $T_\Sigma$, i.e., the terminal locally finite coalgebra. To get a grasp on \emph{all} regular languages, independent of a particular alphabet, we introduce the functor
\[ \rho T: \Set_f^{op} \ra \Cat \]
that maps each finite alphabet $\Sigma$ to $\rho T_\Sigma$, and each letter substitution $h: \Delta\ra \Sigma$ to the morphism $\rho T_\Sigma \ra \rho T_\Delta$ taking preimages under the free monoid morphism $h^*: \Delta^*\ra \Sigma^*$. 

Consider now a variety $V$ of languages in Eilenberg's sense (this is the case $\Cat$ = boolean algebras), and denote by $V\Sigma$ the languages over the alphabet $\Sigma$ contained in $V$. Then the closure condition (a) and the restriction of (c) to \emph{length-preserving} monoid morphisms  (those of the form $h^*$) state precisely that the map $\Sigma \mapsto V\Sigma$ defines a subfunctor \[V \monoto\rho T.\] Closure under left derivatives means that $V\Sigma$ is a subcoalgebra of $\rho T_\Sigma$. Finally, as we will demonstrate below, closure under right derivatives and preimages of arbitrary monoid morphisms amounts categorically to the existence of certain coalgebra homomorphisms. This expresses the closure properties of a variety of languages fully coalgebraically.

In our general setting of two predual categories $\Cat$ and $\DCat$ a \emph{variety of languages in $\Cat$} is thus a subfunctor $V\monoto \rho T$, subject to additional closure properties which are characterized by means of coalgebraic concepts. Dualizing these properties leads to the notion of a \emph{pseudovariety of $\DCat$-monoids}: a class of finite $\DCat$-monoids  closed under finite products, submonoids and homomorphic images. Our main result is the

\medskip
\noindent\textbf{Generalized Eilenberg Theorem.}  Varieties of languages in $\Cat$ correspond bijectively to pseudovarieties of $\DCat$-monoids.
\medskip

All Eilenberg-type theorems mentioned above emerge as special cases by the  corresponding choices of $\Cat$ and $\DCat$.

On our way to proving the Generalized Eilenberg Theorem we will also establish a bijective correspondence between \emph{object-finite} varieties of languages in $\Cat$ (those varieties with $V\Sigma$ finite for all $\Sigma$) and locally
finite varieties (rather than pseudovarieties!) of $\DCat$-monoids. In the case of (ordered) monoids this has been shown by Kl\'ima and Pol\'{a}k \cite{KP}, and to the best of our knowledge it is a new result in all other cases. 

Although our emphasis lies on varieties of \emph{languages}, all our results hold more generally for Moore automata in lieu of acceptors, and hence for varieties of regular behaviors in lieu of regular languages -- one simply replaces the two-element algebra $O_\Cat$ in the above definition of $T_\Sigma$ by an arbitrary finite algebra in $\Cat$.  We briefly explain this at the end of the paper.

\medskip
\noindent\emph{Related Work.} 
Our paper lays a common ground for Eilenberg's original variety theorem \cite{E} and its variants due to Pin~\cite{Pin95}, Pol\'ak~\cite{P} and Reutenauer~\cite{R}.  In our previous paper \cite{GET1} we proved a \emph{local} Eilenberg theorem where one considers classes of regular languages over a fixed alphabet $\Sigma$ versus classes of finite $\Sigma$-generated $\DCat$-monoids in our general setting of predual categories $\Cat$ and $\DCat$. The main technical achievement of \cite{GET1} was the insight that finite subcoalgebras of $\rho T_\Sigma$ closed under right derivatives dualize to finite $\Sigma$-generated $\DCat$-monoids. Our work was inspired by Gehrke, Grigorieff and Pin~\cite{G} who proved a bijective correspondence between  local varieties of languages and classes of finite (ordered) $\Sigma$-generated monoids presented by profinite identities. This result provides a local view of Reiterman's theorem \cite{Rei}  characterizing pseudovarieties of monoids in terms of profinite identities.

Somewhat surprisingly, it has only been in recent years that the fundamental role of duality in algebraic automata theory was fully recognized. Most of the work along these lines concerns the connection between regular languages and profinite algebras. Rhodes and Steinberg~\cite{RS} view the regular languages over $\Sigma$  as a comonoid (rather than just a coalgebra) in the category of boolean algebras, and this comonoid is shown to dualize to the free profinite semigroup on $\Sigma$. Similar results for free profinite monoids can be found in the aforementioned work of Gehrke et al.~which built on previous work of Almeida~\cite{Al} and Pippenger~\cite{Pip97}.

\medskip
\noindent\emph{Acknowledgements.} The authors are grateful to Mai Gehrke, Paul-Andr\'e Melli\`es and Libor Pol\'ak for useful discussions on the topic of our paper. 

\section{Preduality, Monoids and Languages}

In this section we set the scene of predual categories $\Cat$ and $\DCat$, and introduce varieties of languages in $\Cat$ and pseudovarieties of $\DCat$-monoids. The reader is assumed to be familiar with basic category theory and universal algebra.

\subsection{Predual categories}

Our categories of interest are varieties of algebras and varieties of ordered algebras. Given a finitary signature ${\Gamma}$, a \emph{variety of ${\Gamma}$-algebras}
is a full subcategory ${\mathscr A}$ of
$\Alg{\Gamma}$, the category of
${\Gamma}$-algebras and homomorphisms, closed under
quotients (= homomorphic images), subalgebras and products. Equivalently, by Birkhoff's HSP theorem \cite{B1} a variety of algebras is
a class of ${\Gamma}$-algebras specified by equations
$t_1=t_2$ between $\Gamma$-terms. The
forgetful functor ${\mathscr A}\to{\Set}$ of a variety has a
left adjoint assigning to every set $X$ the free algebra over $X$.

Analogously, let $\Alg_\leq{\Gamma}$ be the
category of all \emph{ordered ${\Gamma}$-algebras}. These are $\Gamma$-algebras with a partial order on the underlying set such that all
${\Gamma}$-operations are order-preserving. Morphisms of $\Alg_\leq{\Gamma}$ are
 order-preserving homomorphisms. This category has a factorization
system of surjective homomorphisms and injective order-embeddings.
Thus the concept of a \emph{subalgebra} of an algebra $A$ in $\ACat$ means that the
order is inherited from $A$, whereas a \emph{quotient} of $A$ is
represented by any surjective order-preserving homomorphism with domain $A$. A
\emph{variety of ordered ${\Gamma}$-algebras} is a full
subcategory ${\mathscr A}$ of $\Alg_\leq{\Gamma}$ closed under quotients, subalgebras and products.
Equivalently, by the ordered version of Birkhoff's theorem due to Bloom \cite{Blo}, ${\mathscr A}$ is specified by inequalities
$t_1\le t_2$ between ${\Gamma}$-terms. (Recall that if $t_1$ and $t_2$ lie in $T_\Gamma X$, the discretely ordered algebra of $\Gamma$-terms over variables $X$, then an ordered $\Gamma$-algebra $A$ satisfies $t_1\le t_2$  provided that every
homomorphism $h: T_{\Gamma}X \ra A$ fulfils $h(t_1)\le h(t_2)$.) Again, the forgetful functor
$\ACat\to{\Set}$ has a left adjoint
constructing free algebras.

\begin{XDefinition}\label{def:locfin}
Given a variety ${\mathscr
A}$ of (ordered) algebras we denote by ${\mathscr A}_{ f}$ the full subcategory
of all finite algebras. $\ACat$ is called \emph{locally finite} if all free algebras on finitely many generators are finite.
\end{XDefinition}

In our applications we will encounter the locally finite varieties listed below. We view $\Poset$ as a variety of ordered algebras (over the empty signature) and all other categories as varieties of non-ordered algebras.

\begin{table}[h]\normalsize
\begin{tabularx}{\columnwidth}{lX}
\hline \vspace{-0.2cm}\\
$\Set$         & sets and functions \\
$\PSet$              & pointed sets and point-preserving functions \\
$\Poset$ & partially ordered sets and order-preserving maps \\
$\BA$ &  boolean algebras and boolean homomorphisms\\
$\BR$ & non-unital boolean rings (i.e., non-unital  rings $(R,+,\cdot,0)$ satisfying the equation $x\cdot x = x$) and ring homomorphisms\\
$\DL$ & distributive lattices with $0$ and $1$ and lattice homomorphisms preserving $0$ and $1$\\
$\JSL$ & join-semilattices with $0$ and semilattice homomorphisms preserving $0$\\
$\Vect{K}$ & vector spaces over a finite field $\mathbf{K}$ and linear maps\\
\hline
\end{tabularx}
\end{table}

Here is the central concept for our categorical approach to algebraic automata theory:

\begin{XDefinition}
Two locally finite varieties of (ordered) algebras ${\mathscr C}$ and ${\mathscr D}$ are called
\emph{predual} if their full subcategories ${\mathscr
C}_{ f}$ and ${\mathscr D}_{ f}$ of finite algebras are dually equivalent.
\end{XDefinition}

In what follows $\Cat$ will be a locally finite variety of algebras (where varieties of languages are formed), and $\DCat$ will be a locally finite variety of algebras or ordered algebras (where varieties of monoids are formed). Let us establish some notation for this setting.  The preduality of $\Cat$ and $\DCat$ is witnessed by an equivalence functor $\Cat_f^{op}\xra{\simeq}\DCat_f$ whose action on objects and morphisms we denote by
\[Q\mapsto\widehat Q\quad\text{and}\quad h\mapsto\widehat h.\] The varieties  $\Cat$ and $\DCat$ come equipped with forgetful functors
\[ \under{\mathord{-}}: \Cat\ra\Set\quad\text{and}\quad \under{\mathord{-}}: \DCat\ra \Set.\]
Finally, we write $\one_\Cat$ and $\one_\DCat$ for the free algebras on one generator in $\Cat$ and $\DCat$, respectively, and  $O_\DCat$ and $O_\Cat$ for their dual algebras:
\[ O_\DCat= \widehat{\one_\Cat}  \quad\text{and}\quad \widehat{O_\Cat} \cong \one_\DCat. \]

\begin{XRemark}\label{rem:caniso}The last item is understood as a fixed choice of an algebra $O_\Cat$ in $\Cat_f$ along with an isomorphism $i: \widehat{O_\Cat}\xra{\cong} \one_\DCat$. It follows that $\under{O_\Cat}$ and $\under{O_\DCat}$ are isomorphic:
\begin{align*}
\under{O_\Cat} &\cong \Set(1,\under{O_\Cat}) & (\text{canonically}) \\
&\cong \Cat(\one_\Cat, O_\Cat) & (\text{def. $\one_\Cat$})\\
&\cong \DCat(\widehat{O_\Cat},O_\DCat) & (\text{by duality})\\
&\cong \DCat(\one_\DCat,O_\DCat) & (\text{composition with $i^{-1}$})\\
&\cong \Set(1,\under{O_\DCat}) & (\text{def. $\one_\DCat$})\\
&\cong \under{O_\DCat} & (\text{canonically})
\end{align*}
To simplify the notation we identify the underlying sets of $O_\Cat$ and $O_\DCat$ via this isomorphism and thus assume $\under{O_\Cat}=\under{O_\DCat}$.
\end{XRemark}

\begin
{XAssumptions}\label{A:basic}\ignorespaces For the rest of this paper we fix two predual locally finite varieties $\Cat$ and $\DCat$ with the following properties:
\begin{itemize} 
\item[(i)] ${
\mathscr C}$ is a locally finite variety of algebras.
\item[(ii)] ${
\mathscr D}$ is a locally finite variety of algebras or ordered algebras.
\item[(iii)] Epimorphisms in $\DCat$ are surjective.
\item[(iv)] $\DCat$ is \emph{entropic}: for any two algebras $A$ and $B$ in $\DCat$, the set $[A,B]$  of homomorphisms from $A$ to $B$  is an algebra in $\DCat$ with the pointwise algebraic structure, i.e., a subalgebra of the power $B^{\under{A}} = \prod_{a\in \under{A}} B$.
\item[(v)] $\under{O_\Cat} = \under{O_\DCat} =\{0,1\}$. 
\end{itemize}
\end{XAssumptions}
Condition (iv) means precisely that all operations of the variety $\DCat$ commute, see e.g. \cite[Theorem 3.10.3]{Bor}: given an $m$-ary operation $\sigma$ and an $n$-ary operation $\tau$ in the signature of $\DCat$ and variables $x_{ij}$ ($i=1,\ldots,m$, $j=1,\ldots, n$), the equation
\begin{align*} 
&\sigma(\tau(x_{11},\ldots,x_{1n}),\ldots,\tau(x_{m1},\ldots,x_{mn}))\\
&= \tau(\sigma(x_{11},\ldots, x_{m1}),\ldots, \sigma(x_{1n},\ldots,x_{mn}))
\end{align*}
holds in $\DCat$.
Condition (v) can be lifted to get a theory of regular behaviors in lieu of regular languages, cf. Section \ref{sec:regbeh}.

\begin{XExample}\ignorespaces\label{E:bool}
The following pairs of categories $\Cat$ and $\DCat$ satisfy our assumptions:\\
\begin{center}\vspace{-0.2cm}
\begin{tabular}{llllll}
$\Cat$ & $\BA$ & $\DL$ & $\JSL$ & $\Vect{\Int}$ & $\BR$\\
\hline\vspace{-0.3cm}\\
$\DCat$ & $\Set$ & $\Poset$ & $\JSL$ & $\Vect{\Int}$ & $\PSet$\vspace{0.2cm}\\
\end{tabular}
\end{center}
\noindent(a) $\Cat=\BA$ is predual to $\DCat={\Set}$ via Stone duality \cite{S}: given a finite boolean algebra $Q$, the set $\widehat Q$ consists of all
atoms of $Q$, and given a
homomorphism $h:Q\to R$ in $\BA_f$, the dual function $\widehat h:\widehat R\to\widehat
Q$ in $\Set_f$ is defined by \begin{equation}\label{eq:bool}\widehat
h(r)=\bigwedge\{q\in Q : h(q)\geq r\}.\end{equation}
We choose 
\vspace{-0.2cm}\begin{center}
\begin{table}[h]\normalsize
  \centering
  \renewcommand{\arraystretch}{1.0}
  \begin{tabularx}{\columnwidth}{>{\centering}m{0.1in} >{\centering}m{0.7in} >{\centering}m{0.2in} >{\centering}m{0.2in} >{\centering}m{0.2in} >{\centering}m{0.2in} >{\centering}m{0.2in} >{\centering}m{0.2in}}
    $\one_\Cat =$ & 
    $\xymatrix@=5pt{& \top \ar@{-}[dl] \ar@{-}[dr]&\\ 1 \ar@{-}[dr] && 0 \ar@{-}[dl]\\ & \bot &}$ &
    $\one_\DCat = $  &  
    $\{1\}$ &
    $O_\Cat = $ & 
    $\xymatrix@=5pt{ 1 \ar@{-}[dd]\\\\ 0 }$ &
    $O_\DCat =$ & 
    $\{0,1\}$
  \end{tabularx}
\end{table}
\end{center}\vspace{-0.6cm}
Here $1$ is the generator of the one-generated free algebra $\one_\Cat$ and $\widehat{O_\Cat}=\{1\}=\one_\DCat$. Hence the isomorphisms  $\widehat{O_\Cat}\cong \one_\DCat$ and $\under{O_\Cat}\cong \under{O_\DCat}$ of Remark \ref{rem:caniso} are identity maps.
\endgraf(b)\enspace $\Cat=\DL$ is
predual to $\DCat={\Poset}$ via Birkhoff duality \cite{B2}:  $\widehat Q$ is the poset of
all join-irreducible elements of a finite distributive lattice $Q$ (ordered as in~$Q$), and
every homomorphism $h:Q\to R$ in~$(\DL)_f$ yields a monotone function $\widehat h:\widehat
R\to\widehat Q$ by the formula (1) above. We choose
\vspace{-0.2cm}\begin{center}
\begin{table}[h]\normalsize
  \centering
  \renewcommand{\arraystretch}{1.0}
  \begin{tabularx}{\columnwidth}{>{\centering}m{0.1in} >{\centering}m{0.4in} >{\centering}m{0.2in} >{\centering}m{0.3in} >{\centering}m{0.2in} >{\centering}m{0.2in} >{\centering}m{0.2in} >{\centering}m{0.2in}}
    $\one_\Cat =$ & 
    $\xymatrix@=5pt{0 \ar@{-}[d]\\ 1 \ar@{-}[d] \\ \bot}$ &
    $\one_\DCat = $  &  
    $\{1\}$ &
    $O_\Cat = $ & 
    $\xymatrix@=5pt{ 1 \ar@{-}[dd]\\\\ 0 }$ &
    $O_\DCat =$ & 
    $\xymatrix@=5pt{ 0 \ar@{-}[dd]\\\\ 1 }$ 
  \end{tabularx}
\end{table}
\end{center}\vspace{-0.6cm}
Note that $0$ is the \emph{top} element of $\one_\Cat$ and $O_\DCat$. Again the isomorphisms $\widehat{O_\Cat}\cong \one_\DCat$ and $\under{O_\Cat}\cong \under{O_\DCat}$ are identity maps.
\endgraf(c)\enspace The category~$\Cat=\JSL$ is
self-predual, see \cite{J}, so we can take $\DCat=\JSL$. The dual equivalence associates to each finite semilattice $Q=(X,\vee,0)$ its opposite semilattice $\widehat Q=(X,\wedge,1)$, and to each homomorphism $h:Q\to R$ in $(\JSL)_f$ the homomorphism
$\widehat h: \widehat{Q}\ra\widehat{R}$ with
\begin{equation*}\label{eq:jsl}\widehat
h(r)=\bigvee\{q\in Q : h(q)\leq r\}\end{equation*}
where the join is formed in $Q$. We choose
\vspace{-0.2cm}\begin{center}
\begin{table}[h]\normalsize
  \centering
  \renewcommand{\arraystretch}{1.0}
  \begin{tabularx}{\columnwidth}{>{\centering}m{0.1in} >{\centering}m{0.4in} >{\centering}m{0.2in} >{\centering}m{0.3in} >{\centering}m{0.2in} >{\centering}m{0.2in} >{\centering}m{0.2in} >{\centering}m{0.2in}}
    $\one_\Cat =$ & 
    $\xymatrix@=5pt{0 \ar@{-}[dd]\\\\ 1}$ &
    $\one_\DCat = $  &  
    $\xymatrix@=5pt{ 0 \ar@{-}[dd]\\\\ 1 }$ &
    $O_\Cat = $ & 
    $\xymatrix@=5pt{ 1 \ar@{-}[dd]\\\\ 0 }$ &
    $O_\DCat =$ & 
    $\xymatrix@=5pt{ 1 \ar@{-}[dd]\\\\ 0 }$ 
  \end{tabularx}
\end{table}
\end{center}\vspace{-0.6cm}
The isomorphisms $\widehat{O_\Cat}\cong \one_\DCat$ and $\under{O_\Cat}\cong \under{O_\DCat}$ are  identity maps. Epimorphisms in $\JSL$ are surjective, as proved in \cite{HK}.
\endgraf(d)\enspace The
category~$\Cat=\Vect{\Int}$ of vector spaces over the binary field $\Int=\{0,1\}$ is also
self-predual, so $\DCat=\Vect{\Int}$. The dual equivalence assigns to each finite (i.e.~finite-dimensional) $\Int$-vector space $Q$ its dual space $\widehat Q = [Q,\Int]$. For every
linear map $h:Q\to R$ in $(\Vect{\Int})_f$ the dual map $\widehat h:\widehat R\to\widehat Q$
takes $u:R\to\Int$ to $u\cdot h$. We choose 
\[\one_\Cat = O_\Cat = \Int\quad\text{and}\quad \one_\DCat = O_\DCat = [\Int,\Int].\]
The isomorphism $\widehat{O_\Cat}\cong \one_\DCat$ is identity, and the isomorphism $\under{O_\Cat}\cong\under{O_\DCat}$ identifies the element $1\in \under{O_\Cat}=\under{\Int}$ with $\id \in \under{O_\DCat} = [\Int,\Int]$. Epimorphisms in $\Vect{\Int}$ split and hence are surjective.
\endgraf(e)\enspace An
interesting variation on (a) is the preduality
of $\Cat=\BR$ and $\DCat = \PSet$. Recall that we consider \emph{non-unital} boolean rings (this is Birkhoff's original definition of a boolean ring, see \cite{B}) and homomorphisms preserving $+$, $\cdot$ and $0$. Every \emph{finite} non-unital boolean ring $(Q,+,\cdot, 0)$ can be viewed as a boolean algebra with $x\wedge y = x\cdot y$ and $x\vee y = x + y +x\cdot y$. The preduality  of $\BR$ and $\PSet$ takes $Q$ to the pointed set $\widehat Q = \{\star, q_1,\ldots q_n\}$ where $q_1,\ldots,q_n$ are the atoms of $Q$. A homomorphism $h: Q\ra R$ in $\BR_f$ is mapped to the function $\widehat h: \widehat{R}\ra\widehat{Q}$ with $\widehat h(r)$ defined by (1) if $r\neq\star$ and some $q\in Q$ with
$h(q)\geq r$ exists, and otherwise $\widehat{h}(r)=\star$. We choose
\[ \one_\Cat = O_\Cat = \{0,1\} \quad\text{and}\quad \one_\DCat = O_\DCat = \{\star,1\}. \]
The isomorphism $\widehat{O_\Cat}\cong \one_\DCat$ is identity, and the isomorphism $\under{O_\Cat}\cong \under{O_\DCat}$ identifies $0\in\under{O_\Cat}$ with $\star\in \under{O_\DCat}$.
\end{XExample}

The preduality of  $\Cat$ and $\DCat$ allows us to model deterministic automata both as coalgebras and as algebras for suitable endofunctors on $\Cat$ and $\DCat$. Fix a finite input alphabet $\Sigma$ and consider first the endofunctor on $\Cat$
\[ T_\Sigma Q = O_\Cat \times Q^\Sigma = O_\Cat \times \prod_\Sigma Q.\]
A \emph{$T_\Sigma$-coalgebra} $(Q,\gamma)$ consists of an object $Q$ in $\Cat$ together with a morphism $\gamma: Q\ra T_\Sigma Q$. By the universal property of the product, to give a $T_\Sigma$-coalgebra  means to give  morphisms $\gamma_a: Q\ra Q$ for each $a\in \Sigma$, respresenting transitions, and a morphism $\gamma_{\out}: Q\ra O_\Cat$ representing final states. Hence $T_\Sigma$-coalgebras are deterministic $\Sigma$-automata in $\Cat$ without an initial state, often denoted as triples
\[ Q=(Q,\gamma_a,\gamma_{\out}).\]
 A \emph{homomorphism} $h:(Q,\gamma')\ra(Q',\gamma')$ of $T_\Sigma$-coalgebras is a morphism $h: Q\ra Q'$ in $\Cat$ with $\gamma'\cdot h = T_\Sigma h\cdot \gamma$, which means that  the following diagram commutes for all $a\in\Sigma$:
\[
\xymatrix{
Q \ar[r]^{\gamma_a} \ar[d]_h & Q \ar[d]^h \ar[dr]^{\gamma_{\out}} & \\
Q' \ar[r]_{\gamma_{a}'} & Q' \ar[r]_{\gamma_{\out}'} & O_\Cat 
}
\]
We write
\[ \Coalg{T_\Sigma} \quad\text{and}\quad \FCoalg{T_\Sigma}\]
for the categories of (finite) $T_\Sigma$-coalgebras and homomorphisms.
The functor $T_\Sigma$ has an associated endofunctor on $\DCat$
\[ L_\Sigma A = \one_\DCat + \coprod_\Sigma A\]
which is \emph{predual} to $T_\Sigma$ in the sense that the restrictions
\[T_\Sigma: \Cat_f\ra\Cat_f \quad\text{and}\quad L_\Sigma: \DCat_f \ra \DCat_f, \]
to finite algebras are dual. That is, the diagram below commutes up to natural isomorphism:
\[
\xymatrix{
\Cat_f^{op} \ar[d]_{\widehat{(\mathord{-})}} \ar[r]^{T_\Sigma^{op}} & \Cat_f^{op} \ar[d]^{\widehat{(\mathord{-})}} \\
\DCat_f \ar[r]_{L_\Sigma}  & \DCat_f  
}
\]
Dually to the concept of a $T_\Sigma$-coalgebra, an \emph{$L_\Sigma$-algebra} $(A,\alpha)$ is an object $A$ in $\DCat$ together with  a morphism $\alpha: L_\Sigma A \ra A$. Equivalently, an $L_\Sigma$-algebra is given by morphisms $\alpha_a: A\ra A$ for $a\in\Sigma$, representing transitions, and a morphism $\alpha_{\init}: \one_\DCat\ra A$ selecting an initial state. $L_\Sigma$-algebras are denoted as
\[ A = (A,\alpha_a, \alpha_{\init}) \]
and can be viewed as deterministic $\Sigma$-automata in $\DCat$ without final states. A \emph{homomorphism} $h: (A,\alpha)\ra (A',\alpha')$ of $L_\Sigma$-algebras is a $\DCat$-morphism $h: A\ra A'$ with $h\cdot \alpha = \alpha'\cdot L_\Sigma h$. We write 
\[ \Alg{L_\Sigma} \quad\text{and}\quad \FAlg{L_\Sigma}\]
for the categories of (finite) $L_\Sigma$-algebras and homomorphisms. From the preduality of $T_\Sigma$ and $L_\Sigma$ we immediately conclude:
 
\begin{XLemma}
$\FCoalg{T_\Sigma}$ and $\FAlg{L_\Sigma}$ are dually equivalent categories. The dual equivalence is given on objects and morphisms by
\[ (Q,\gamma) \mapsto (\widehat Q, \widehat{\gamma})\quad\text{and}\quad h\mapsto \widehat{h}.\]
\end{XLemma}

In triple notation this dual equivalence maps a finite $T_\Sigma$-coalgebra $(Q,\gamma_a,\gamma_{\out})$ to its \emph{dual $L_\Sigma$-algebra} $(\widehat{Q},\widehat{\gamma_a},\widehat{\gamma_{\out}})$. Hence finite $T_\Sigma$-coalgebras and finite $L_\Sigma$-algebras are essentially the same structures, and both can be understood as finite automata with additional (ordered) algebraic structure.
\begin
{XExample}(a) Let $\Cat=\BA$ and $\DCat=\Set$. Then ${L_{\Sigma}}$-algebras are the usual concept of a deterministic automaton without final states. A $T_\Sigma$-coalgebra is a deterministic $\Sigma$-automaton with a boolean algebra structure on the state set $Q$ such that (i)~all transitions $\gamma_a:Q\to Q$ are boolean
homomorphisms and (ii)~the final states form an ultrafilter, determined by the preimage of $1$ under the morphism $\gamma_{\out}: Q\ra O_\Cat=\{0,1\}$. For
\emph{finite} automata (ii)\ means that precisely one atom $i\in Q$ is final
and all final states form the upper set $\mathop\uparrow i = \{q\in Q: q\geq i\}$. The dual equivalence of the previous lemma takes a finite boolean automaton $Q=(Q,\gamma_a,\gamma_{\out})$ to the automaton $\widehat{Q}=(\widehat Q, \widehat{\gamma_a},\widehat{\gamma_{\out}})$ in $\Set$ whose states are the atoms of $Q$. The unique atomic final state $i\in Q$ is the initial state of $\widehat{Q}$, and there is a transition $x \overset{a}{\ra} x'$  in $\widehat{Q}$ iff 
\[ x' = \bigwedge\{ y: y  \overset{a}{\ra} y' \text{ in $Q$ for some $y'\geq x$} \}.\]

\endgraf(b)\enspace The case $\Cat=\DL$ and $\DCat=\Poset$ is analogous. Here $L_\Sigma$-algebras are \emph{ordered} deterministic automata without final states. $T_\Sigma$-coalgebras carry a distributive lattice structure on $Q$, and again for finite $Q$ the final states form an upper set $\mathop\uparrow i$.
 \end{XExample}
 
 \begin{XRemark}\label{R:fact}
We frequently need to factorize (co-)algebra homomorphisms into a surjective and an injective part. For the variety $\Cat$ we choose the factorization system
 \[({\text{strong epi}},{\text{mono}})=({\text{surjective}},{\text{injective}}).\] Recall that an
 epimorphism~$e$ is called strong if it has the diagonal fill-in
 property w.\thinspace r.\thinspace t.\ all monomorphisms~$m$: given morphisms $u$,~$r$ with $r\cdot
 e=m\cdot u$ there exists $d$ with $u=d\cdot e$. Since the
 functor ${T_{\Sigma}}Q={O_{\mathscr C}}\times
 Q^{\Sigma}$ preserves monomorphisms, this
 factorization system of ${\mathscr C}$ lifts to
 $\Coalg{T_\Sigma}$: every
coalgebra homomorphism $h:{(Q,\gamma)}\to{(Q',
 \gamma')}$ factorizes as a \emph{quotient coalgebra} of $(Q,\gamma)$ followed by
 a \emph{subcoalgebra} of ${(Q',\gamma')}$, by which we mean that the underlying morphism in $\Cat$ is a strong epimorphism (or a monomorphism, respectively).
 
 Dually, in~${\mathscr D}$ we consider the factorization system
  \[(\text{epi, strong mono}).\]  If $\DCat$ is a variety of algebras, this is just the (surjective, injective)-factorization system as epimorphisms in $\DCat$ are surjective by Assumption \ref{A:basic}(iii). In  case  $\DCat$ is a variety of ordered algebras, we get the (surjective, injective order-embedding)-factorization system, which is clearly the ``right'' one for ordered algebras. Since 
 ${L_{\Sigma}}=\one_\DCat + \coprod_\Sigma A$ preserves epimorphisms,
 the factorization system of ${\mathscr D}$ lifts to
 $\Alg{L_\Sigma}$: every $L_\Sigma$-algebra homomorphism $h: (A,\alpha)\ra (A',\alpha')$ factorizes as a \emph{quotient algebra} of $(A,\alpha)$ followed by a \emph{subalgebra} of $(A',\alpha')$, i.e., the underlying morphism in $\DCat$ is an epimorphism (or a strong monomorphism, respectively).
 \end{XRemark}
 
\subsection{${\mathscr D}$-Monoids}

Our Assumption \ref{A:basic}(iv) that $\DCat$ be entropic admits a more  categorical interpretation. Given objects $A$, $B$ and $C$ in $\DCat$, a \emph{bimorphism} is a function $f: \under{A}\times \under{B} \ra \under{C}$ such that $f(a,\mathord{-}): B\ra C$ and $f(\mathord{-},b): A\ra C$ are morphisms of $\DCat$ for all $a\in A$ and $b\in B$. A \emph{tensor product} of $A$ and $B$ is a universal bimorphism $t: \under{A}\times \under{B} \ra \under{A\otimes B}$, i.e., for every bimorphism $f: \under{A}\times \under{B} \ra \under{C}$ there is a unique morphism $f'$ in $\DCat$ making the diagram below commute.
\[
\xymatrix{
\under{A}\times \under{B} \ar[dr]_f \ar[r]^t & \under{A\otimes B} \ar@{-->}[d]^{\under{f'}}\\
& \under{C}
}
\]
As shown in \cite{BN}, tensor products exist in any variety of (ordered) algebras, and ``entropic'' means precisely that $\DCat= (\DCat,\otimes,\one_\DCat)$ is a symmetric monoidal closed category. In particular, as in any monoidal category, we have a notion of \emph{monoid} in $\DCat$. In our setting this means the following:

\begin{XDefinition}\ignorespaces\label
{D:mon} A \emph{${\mathscr D}$-monoid} $(D,{\circ},i)$ consists of an object $D$ of ${\mathscr D}$ and a monoid structure $(\under{D},{\circ},i)$ on its underlying set $\mathopen|D\mathclose|$ whose multiplication is a bimorphism, i.e., for every $x\in\mathopen|D\mathclose|$ both $x\circ{\mskip1.5mu \char"7B\mskip1.5mu }$ and ${\mskip1.5mu \char"7B\mskip1.5mu }\circ x$ are endomorphisms of~$D$ in $\DCat$. A \emph{morphism} $h: (D,{\circ},i)\ra (D',{\circ'},i')$ of $\DCat$-monoids is a morphism $h: D\ra D'$ in $\DCat$ preserving the monoid structure. The $\DCat$-monoids and their morphisms form a category
\[\Mon{\DCat}.\] 
\end{XDefinition}
Observe that  $\Mon{\DCat}$  is a variety of (ordered)
algebras: add $\circ$ and $i$ to the  signature of $\DCat$, and add the monoid axioms and equations expressing that $\circ$ is a bimorphism to the (in)equalities presenting~${\mathscr
D}$. The factorization system of $\DCat$, see Remark \ref{R:fact}, lifts to $\Mon{\DCat}$. Hence a \emph{submonoid} of a ${\mathscr D}$-monoid $D$ is a $\DCat$-monoid morphism into $D$ carried by a strong monomorphism in~${\mathscr D}$, and a \emph{quotient monoid} of $D$ is a $\DCat$-monoid morphism with domain $D$ carried by an epimorphism in~${\mathscr D}$. 

\begin{XExamples} For our categories $\DCat$ of Example \ref{E:bool} the $\DCat$-monoids are characterized as follows:

(a) $\DCat=\Set$:  (ordinary) monoids.

(b) $\DCat=\Poset$: ordered monoids.

(c) ${\mathscr
D}=\JSL$: idempotent semirings, i.e., semirings $(S,{+},{\cdot},0)$ satisfying the equation $x+x=x$.
Indeed, this means precisely that $(S,{+},0)$ is a semilattice,
and the distributive laws of a semiring express that the multiplication
preserves $+$ and $0$ in each variable, so that it
is a bimorphism. 

(d) $\DCat=\Vect{\Int}$: $\Int$-algebras in the classical sense of algebras
over a field, i.e., a $\Int$-algebra is a vector space over $\Int$ together with a monoid structure whose multiplication is distributive (=
linear in each variable). 

(e) $\DCat=\PSet$: monoids with $0$, i.e., monoids containing an element $0$ such that $x\circ 0 = 0\circ x = 0$ for all $x$. Morphisms of $\Mon{\DCat}$ are monoid morphisms preserving $0$.
\end{XExamples}

\begin{XDefinition}\label{def:pseudovar}
A \emph{pseudovariety of $\DCat$-monoids} is a class of finite $\DCat$-monoids closed under submonoids, quotients and finite products.
\end{XDefinition}

On our way to proving the generalized Eilenberg theorem we will encounter monoids with a specified set of generators. To this end we need to describe the free $\DCat$-monoids. We write \[\Psi: \Set\ra\DCat\] for the free algebra functor, that is, the left adjoint of the forgetful functor $\under{\mathord{-}}: \DCat\ra\Set$. For notational simplicity we assume that $X$ is a subset of $\under{\Psi X}$, and the universal map is the inclusion $X\monoto \under{\Psi X}$. (For nontrivial varieties $\DCat$ of (ordered) algebras such a choice of $\Psi X$ is always possible.)

\begin{XProposition}[see~\cite{GET1}, Prop.~4.22]\label{P:free} The free ${\mathscr D}$-monoid on
a finite set ${\Sigma}$ is the ${\mathscr D}$-monoid
$(\Psi{\Sigma}^*,{\bullet},\varepsilon)$ with multiplication $\bullet$ extending
the concatenation of words over ${\Sigma}$, and universal map $\Sigma\monoto \Sigma^* \monoto \under{\Psi\Sigma^*}$. That is, every map $f:\Sigma\ra \under{D}$ into a $\DCat$-monoid $D$ extends uniquely to a $\DCat$-monoid morphism $\ol f$:
\[
\xymatrix{
\Sigma \ar@{>->}[r]\ar[dr]_f & \under{\Psi\Sigma^*} \ar@{-->}[d]^{\under{\overline f}}\\
& \under{D}
}
\]
\end{XProposition}
\begin{XExamples} In our categories $\DCat$ of Example \ref{E:bool} the free $\DCat$-monoids $\Psi\Sigma^*$ on $\Sigma$ are characterized as follows:

(a) ${\mathscr D}={\Set}$: the usual free monoid $\Sigma^*$.

(b)  ${\mathscr D}={\Poset}$: the free ordered monoid $\Sigma^*$ (discretely ordered).

(c) $\DCat=\JSL$: the free idempotent semiring $\Pow_f \Sigma^*$, carried by the set of  finite languages over $\Sigma$. The semilattice structure is union and $\emptyset$, the monoid
multiplication is the concatenation $L_1\bullet L_2=L_1L_2$ of
 languages, and the monoid unit is $\{\epsilon\}$.
 
(d) $\DCat=\Vect{\Int}$: the free $\Int$-algebra $\Pow_f\Sigma^*$. Its vector addition is the symmetric
difference $L_1\oplus L_2=(L_1\setminus L_2)\cup(L_2\setminus L_1)$, and the zero
vector is $\emptyset$. The monoid unit is again $\{\epsilon\}$, and the monoid multiplication is  $\Int$-weighted
concatenation of languages: $L_1\bullet L_2$ consists of all words $w$ having an odd number of decompositions $w=w_1w_2$ with $w_i\in L_i$.

(e) $\DCat=\PSet$: the free monoid with $0$. This is the monoid $\Sigma^* + \{0\}$ arising from $\Sigma^*$ by adding a zero element.
\end{XExamples}

\begin{XDefinition}
(a) A ${\mathscr D}$-monoid is called \emph{${\Sigma}$-generated} if a set of generators indexed by ${\Sigma}$ is given in it. Equivalently: if it is a quotient monoid of the free ${\mathscr D}$-monoid $\Psi{\Sigma}^*$. \emph{Morphisms} of ${\Sigma}$-generated ${\mathscr D}$-monoids are required to preserve the given generators. That is, given two ${\Sigma}$-generated ${\mathscr D}$-monoids $e_k:\Psi{\Sigma}^*\onto(D_k,{\circ_k},i_k)$, $k=1$,~$2$, a morphism of ${\Sigma}$-generated ${\mathscr D}$-monoids  is a $\DCat$-monoid morphism $f:(D_1,{\circ_1},i_1)\to(D_2,{\circ_2},i_2)$ with $e_2=f\cdot e_1$.

(b) The \emph{subdirect product} of two ${\Sigma}$-generated ${\mathscr D}$-monoids $e_k:\Psi{\Sigma}^*\epito(D_k,{\circ_k},i_k)$ is the $\DCat$-submonoid of their product which is the image of $\langle e_1,e_2\rangle:\Psi{\Sigma}^*\to D_1\times D_2$.
\end{XDefinition}

\begin{XDefinition}
By a \emph{pseudovariety of ${\Sigma}$-generated ${\mathscr D}$-monoids} is meant a collection of finite ${\Sigma}$-generated ${\mathscr D}$-monoids closed under subdirect products and quotients.
\end{XDefinition} 
In other words, if ${\mathcal L}$ denotes the poset of all finite quotients of $\Psi{\Sigma}^*$ in $\Mon{\DCat}$, then a pseudovariety is a subposet closed under finite joins ($=$ subdirect products) and closed downwards (i.e., under quotients). Here we use the ordering of quotients $e:\Psi{\Sigma}^*\epito D$ where $e_1\le e_2$ iff $e_1$ factorizes through~$e_2$.

\begin{XRemark}\label{rem:assalg}  Every ${\Sigma}$-generated ${\mathscr
D}$-monoid $e:\Psi{\Sigma}^*\onto(D,{\circ},i)$ defines the
\emph{associated $L_\Sigma$-algebra}
$\alpha:{L_{\Sigma}}D\to D$ with the same object $D$ of states, initial state
$e(\varepsilon)$ and  transitions given by right
multiplication $\alpha_a(d)= d \circ e(a)$ for $a\in\Sigma$. In particular, the associated $L_\Sigma$-algebra of the free ${\mathscr D}$-monoid
$\Psi{\Sigma}^*$ has the initial state $\varepsilon$ and
the transitions $\mathord{-}\bullet a$ for
$a\in{\Sigma}$.  This means that the above $L_\Sigma$-algebra structure of $D$ is the unique one that makes $e$ a homomorphism of $L_\Sigma$-algebras.

As shown in \cite[Prop.~4.29]{GET1}, $\Psi\Sigma^*$ is the initial $L_\Sigma$-algebra: for every ${L_{\Sigma}}$-algebra
$(A,\alpha)$ there exists a unique $L_\Sigma$-algebra homomorphism
\[e_A:\Psi{\Sigma}^*\to A.\] 
Its restriction to $\Sigma^*$ computes the action of $A$ on words $w\in{\Sigma}^*$:
\[ e_A(w) = \alpha_w\cdot \alpha_\init: \one_\DCat \ra A. \]
 Here we use the notation \[ \alpha_w={\alpha_{a_n}}\cdots{\alpha_{a_1}}:A
\to A\qquad\hbox{for $w=a_{1}\cdots a_{n}$}.\]
Analogously, for coalgebras $(Q,\gamma)$ we put $\gamma_w= {\gamma_{a_n}}\cdots{\gamma_{a_1}}$.
\end{XRemark}

Since $\Sigma$-generated monoids are $L_\Sigma$-algebras one may ask for the converse: given an $L_\Sigma$-algebra, is it associated to some $\Sigma$-generated monoid? In the next subsection we will see that this question is, by duality, directly related to closure properties of classes of regular languages.

 \subsection{Languages}\label{2.3}

Our categorical approach to varieties of languages starts with a characterization of the regular languages over a fixed alphabet $\Sigma$ by a universal property. Let us call a $T_\Sigma$-coalgebra $Q$ \emph{locally finite} if it is a filtered colimit of finite coalgebras. Or equivalently, if every state $q\in Q$ lies in a finite subcoalgebra of $Q$. As shown in \cite{M}, the terminal locally finite coalgebra $\rho T_\Sigma$ -- characterized by the property that every locally finite coalgebra has a unique homomorphism into it -- is the filtered colimit of the diagram 
\[ \FCoalg{T_\Sigma} \monoto \Coalg{T_\Sigma}  \]
of \emph{all} finite coalgebras. Its coalgebra structure is an isomorphism
\[ \rho T_\Sigma \cong T_\Sigma(\rho T_\Sigma), \]
which is why $\rho T_\Sigma$ is also called the \emph{rational fixpoint} of $T_\Sigma$. 

\begin{XProposition}[see~\cite{GET1}, Cor.~2.11]\ignorespaces\label{E:reg} $\rho T_\Sigma$ is carried by the set of all regular languages over $\Sigma$. The transition morphisms are carried by left derivatives $L\mapsto a^{-1}L$ for all $a\in\Sigma$, and the final states are precisely the languages containing the empty set. For any locally finite coalgebra $(Q,\gamma)$ the unique homomorphism $L_Q: Q\ra \rho T_\Sigma$ maps a state $q\in Q$ to the \emph{language accepted by $q$}:
\[ L_Q(q) = \{w\in\Sigma^*: \gamma_{\out}\cdot \gamma_w(q) = 1\}. \]
\end{XProposition}

\begin{XExample}\label{ex:rhot}
Continuing our Examples~\ref{E:bool}, $\rho T_\Sigma$ has the following algebraic structure as an object of the variety $\Cat$:

(a) For $\Cat=\BA$ the boolean algebra structure  is $\cup$, $\cap$, $\overline{(\mathord{-})}$, $\emptyset$ and $\Sigma^*$.

(b) For $\Cat=\DL$ the lattice structure is $\cup$, $\cap$, $\emptyset$ and $\Sigma^*$.

(c) For $\Cat=\JSL$ the semilattice structure is $\cup$ and $\emptyset$.

(d) For $\Cat=\Vect{\Int}$ the vector addition is symmetric difference $L\oplus M = (L\setminus M) \cup (M\setminus L)$ and the zero vector is $\emptyset$.

(e) For $\Cat=\BR$ the multiplication is $\cap$, the addition is symmetric difference $\oplus$, and the zero element is $\emptyset$.
\end{XExample}

\begin{XDefinition}\ignorespaces\label{D:loc} By a \emph{local variety
of languages} over ${\Sigma}$ in ${\mathscr C}$ is meant a
subcoalgebra $Q\hookrightarrow\rho{T_{\Sigma}}$ of the
rational fixpoint closed under right derivatives,
i.e., $L\in\mathopen|Q\mathclose|$ implies $La^{-1}\in\under{Q}$ for all $a\in\Sigma$.
\end{XDefinition}

Note that a local variety is closed under left derivatives
automatically, being a subcoalgebra of $\rho{T_{\Sigma}}$.
Closure under right derivatives also admits a fully coalgebraic description:

\begin{XNotation}\label{not:der}
Given a $T_\Sigma$-coalgebra ${(Q,\gamma)}$ and an input $a\in{\Sigma}$, denote by ${(Q,\gamma)}_a$ the $T_\Sigma$-coalgebra with the same states and transitions, but whose final-state morphism is ${\gamma_{\out}}\cdot\gamma_a:Q\to{O_{\mathscr C}}$. 
\end{XNotation}

\begin{XProposition}[see \cite{GET1}, Prop.~4.3]\label{P:right} A
subcoalgebra ${(Q,\gamma)}$ of $\rho{T_{\Sigma}}$ is a
local variety iff a $T_\Sigma$-coalgebra homomorphism
from ${(Q,\gamma)}_a$ to ${(Q,\gamma)}$ exists for every
$a\in{\Sigma}$. \end{XProposition}

\begin{XProposition}[see \cite{GET1}, Prop.~3.24 and 4.32]\label{prop:rqcmon} A finite $T_\Sigma$-coalgebra $Q$ is a subcoalgebra of $\rho T_\Sigma$ iff its dual $L_\Sigma$-algebra $\widehat Q$ is a quotient algebra of $\Psi\Sigma^*$. In this case, $Q$ is
a local variety iff $\widehat{Q}$ is the associated $L_\Sigma$-algebra of some finite ${\Sigma}$-generated ${\mathscr D}$-monoid. \end{XProposition}

In other words, for any finite coalgebra $Q$ the unique $T_\Sigma$-coalgebra homomorphism $L_Q: Q\ra \rho T_\Sigma$ of Proposition \ref{E:reg} is injective iff the unique $L_\Sigma$-algebra homomorphism $e_{\widehat Q}: \Psi\Sigma^*\ra \widehat{Q}$ of Remark~\ref{rem:assalg} is surjective. Moreover, if $Q\monoto \rho T_\Sigma$ is a local variety of languages, there exists a (unique) monoid structure on $\widehat{Q}$ making $e_{\widehat Q}$ a $\DCat$-monoid morphism. In this case we call $\widehat{Q}$ the \emph{dual $\Sigma$-generated $\DCat$-monoid of $Q$}.

Proposition \ref{prop:rqcmon} was the basis of the main result of \cite{GET1}. Observe that the set of all local varieties of languages over $\Sigma$ in $\Cat$ forms a complete lattice whose meet is intersection. Analogously for the set of all pseudovarieties of $\Sigma$-generated $\DCat$-monoids.

\begin{XTheorem}[Local Eilenberg Theorem~\cite{GET1}, Thm.~4.36] The lattice of all local varieties of
languages over $\Sigma$ in $\Cat$ is isomorphic to the
lattice of all pseudovarieties of ${\Sigma}$-generated
${\mathscr D}$-monoids. \end{XTheorem}

Local varieties of languages are local in the sense that a fixed alphabet $\Sigma$ is considered. To get a global (alphabet-independent) view of all regular languages, we extend the map $\Sigma \mapsto \rho T_\Sigma$ to a functor $\rho T:{\Set}_{
f}^{ op}\to{\mathscr C}$. Observe first that for every
$h:{\Delta}\to{\Sigma}$ in $\Set_f$  there is a morphism $Q^h: Q^\Sigma \ra Q^\Delta$ given by precomposition with $h$. Hence we can turn each
$T_\Sigma$-coalgebra $(Q,\gamma)$ into the $T_\Delta$-coalgebra $(Q,\gamma)^h$ with the same states $Q$ and coalgebra structure
\[
  Q \xra{\gamma}{O_{\mathscr C}}\times Q^{\Sigma} \xra{\id\times Q^h} {O_{\mathscr C}}\times Q^{\Delta}. 
\]
This is a familiar construction for deterministic automata: if some state $q\in Q$ accepts the language $L\seq \Sigma^*$ in $(Q,\gamma)$, then it accepts the language $(h^*)^{-1}(L)\seq \Delta^*$ in $(Q,\gamma)^h$.  Here $h^*: \Delta^*\ra\Sigma^*$ denotes the free extension of $h$ to a monoid morphism. 

The coalgebra $(Q,\gamma)^h$ is locally finite if $(Q,\gamma)$ is,  as every subcoalgebra of $(Q,\gamma)$ is also a subcoalgebra of $(Q,\gamma)^h$. In particular, $(\rho T_\Sigma)^h$ is locally finite, so there is a unique $T_\Delta$-coalgebra homomorphism $\rho T_h: (\rho T_\Sigma)^h \to \rho T_\Delta$ into the terminal locally finite $T_\Delta$-coalgebra. The morphism $\rho T_h$ forms preimages under the monoid morphism  $h^*$: 
\[\rho T_h(L)=(h^*)^{-1}(L)\qquad\hbox {for all $L\in \under{\rho T_\Sigma}$}.\]

\begin{XDefinition}\label{def:ratfunc}\ignorespaces\label
{D:rat} The \emph{rational functor} $\rho T:{\Set}_{
f}^{ op}\to{\mathscr C}$ assigns to every finite alphabet
${\Sigma}$ the rational fixpoint $\rho{T_{\Sigma}}$
and to every map $h:{\Delta}\to{\Sigma}$ the morphism $\rho T_h: \rho T_\Sigma\ra\rho T_\Delta$. \end{XDefinition}

In the classical case  ${\mathscr C}=\BA$, the rational functor maps each finite alphabet ${\Sigma}$ to the boolean algebra of regular languages over $\Sigma$. A
variety $V$ of languages in Eilenberg's sense (see Introduction) is thus a
subfunctor\footnote{Recall that a \emph{subfunctor} of a functor $F:\ACat\ra\BCat$ is a natural transformation $m: F'\monoto F$ with monomorphic components $m_A: F'A\monoto FA$. To specify $F'$ is suffices to give the object map $A\mapsto F'A$ and monomorphisms $m_A$ such that, for each $f: A\ra A'$ in $\ACat$, the morphism $Ff\cdot m_A$ factorizes through $m_{A'}$. This uniquely determines the action of $F'$ on morphisms.} 
$V\monoto\rho T$
that assigns to
every finite alphabet $\Sigma$ a local variety $V\Sigma\monoto \rho T_\Sigma$, and is closed under preimages
of monoid morphisms
$f:{\Delta}^*\to{\Sigma}^*$. To formulate the preimage
condition categorically, we identify any language $L\seq\Sigma^*$ with its characteristic function $L:{\Sigma}^*\to{\{0,1\}}$ in $\DCat=\Set$. Then the preimage of $L$ under
$f$ is precisely the language represented by the composite
function $L\cdot f:{\Delta}^*\to{\{0,1\}}$. Thus the
missing condition on our subfunctor $V$ is the following: for
every language $L\in V{\Sigma}$ we have $L\cdot f\in
V{\Delta}$. Let us now extend these considerations to our general setting. 

\begin{XNotation}\label{N:OD}
Using the adjunction $\Psi \dashv \under{\mathord{-}}: \DCat\ra \Set$, we identify any language $L:{\Sigma}^*\to{\{0,1\}=\under{O_\DCat}
}$
with the corresponding morphism
$L:\Psi{\Sigma}^*\to{O_{\mathscr D}}$ of~$\DCat$.  The \emph{preimage} of $L$ under a
${\mathscr D}$-monoid morphism $f:\Psi{\Delta}^*\to\Psi
{\Sigma}^*$ is the language $L\cdot f: \Psi\Delta^*\ra O_\DCat$ over the alphabet
${\Delta}$. 
\end{XNotation}

\begin{XDefinition}\ignorespaces\label
{D:var} (a) A subfunctor $V\monoto
\rho{T}$ of the rational functor is \emph{closed
under preimages} if, for every ${\mathscr D}$-monoid morphism
$f:\Psi{\Delta}^*\to\Psi{\Sigma}^*$ and every language
$L:\Psi{\Sigma}^*\to{O_{\mathscr D}}$ in $V{\Sigma}$,
the language $L\cdot f$ lies in $V\Delta$.

(b) By a \emph{variety of languages} in ${\mathscr C}$ is meant a
subfunctor $V \monoto \rho T$ closed under preimages
 such that $V{\Sigma}$ is a local variety of languages for every
alphabet ${\Sigma}\in{\Set}_{ f}$.
\end{XDefinition}

In Theorem~\ref{T:pre} below we give a fully
coalgebraic characterization of preimage closure.

\begin{XExamples}(a)~The case ${\mathscr C}=\BA$ and $\DCat=\Set$ captures
the original concept of Eilenberg~\cite{E}: a variety of languages in $\BA$ forms boolean subalgebras $V\Sigma$ of $\rho T_\Sigma$, closed under derivatives and preimages of monoid morphisms $f:\Delta^*\ra \Sigma^*$.
\endgraf(b)\enspace In the case
$\Cat=\DL$ and $\DCat=\Poset$ we just drop closure under complement: a variety of languages in $\DL$ forms
sublattices $V{\Sigma}$ of
$\rho T_\Sigma$ closed under derivatives and preimages of monoid morphisms $f:\Delta^*\ra \Sigma^*$. This is the concept of a \emph{positive variety of languages} studied by Pin \cite{Pin95}.
\endgraf(c)\enspace Let $\Cat=\DCat=\JSL$. Given a language
$L\subseteq{\Sigma}^*$, the corresponding semilattice
morphism $L:{\mathcal P}_{ f}{\Sigma}^*\to{\{0,1\}}$
takes a finite language $\{w_{1},\ldots,w_{k}\}$ to $1$ iff
$w_i\in L$ for some~$i$ (this follows from
$\{w_{1},\ldots,w_{k}\}=\bigvee_{i=1}^k\{w_i\}$). The preimage
of $L$ under a semiring morphism $f:{\mathcal P}_{
f}{\Delta}^*\to{\mathcal P}_{ f}{\Sigma}^*$  corresponds to the language $M\seq\Delta^*$ of all words $u\in\Delta^*$ for which $f(u)$ contains some word of $L$. A variety of languages in $\JSL$ forms subsemilattices $V\Sigma$ of $\rho T_\Sigma$ closed under derivatives and preimages of semiring morphisms $f:{\mathcal P}_{
f}{\Delta}^*\to{\mathcal P}_{ f}{\Sigma}^*$ . This is
the notion of variety introduced by Pol\'ak~\cite{P}.

\endgraf(d)\enspace If ${\mathscr C}=\DCat=\Vect{\Int}$, the linear map $L:\Pow_f\Sigma^* \ra \{0,1\}$ corresponding to $L\seq \Sigma^*$ takes
$\{w_{1},\ldots,w_{k}\}\in{\mathcal P}_{ f}{\Sigma}^*$
to $1$ iff $w_i\in L$ for an odd number of $i=1, \ldots, k$.
Thus the preimage of $L$ under a $\Int$-algebra morphism $f:{\mathcal P}_{
f}{\Delta}^*\to{\mathcal P}_{ f}{\Sigma}^*$ corresponds to the language $M\seq\Delta^*$ of all
words $u\in\Delta^*$ for which $f(u)$ contains an odd number of words of $L$. A variety of languages in $\Vect{\Int}$ forms linear subspaces $V\Sigma$ of $\rho T_\Sigma$ closed under derivatives and preimages of $\Int$-algebra morphisms $f:{\mathcal P}_{
f}{\Delta}^*\to{\mathcal P}_{ f}{\Sigma}^*$ . This notion of a variety was introduced by Reutenauer~\cite{R}; see also Section \ref{sec:regbeh}.

(e) Finally, let $\Cat=\BR$ and $\DCat=\PSet$. The preimage of $L\seq \Sigma^*$ under a zero-preserving monoid morphism $f: \Delta^*+\{0\} \ra \Sigma^*+\{0\}$ consists of all words $w\in\Delta^*$ for which $f(w)$ lies in $L$. A variety of languages in $\BR$ forms subrings $V\Sigma$ of $\rho T_\Sigma$ closed under derivatives and preimages of zero-preserving monoid morphisms.
\end{XExamples}

 See the table in Section \ref{S:GET} for a summary of our examples. The set of all varieties of languages in ${\mathscr C}$ is a complete lattice since any intersection of varieties (formed objectwise) is a variety. The same holds for the set of all pseudovarieties of $\DCat$-monoids. Our main result, the \emph{Generalized Eilenberg Theorem} (see Theorem~\ref{T:Eil}), states that these two lattices are isomorphic. The rest of the paper is devoted to the proof.

\section{Coalgebraic and Algebraic Language Acceptance}

In this section we compare the languages accepted by a finite $T_\Sigma$-coalgebra in ${\mathscr C}$ with those accepted by its dual finite $L_\Sigma$-algebra in~${\mathscr D}$.

\begin{XNotation}\label{not:acclang}
(a) Recall $\under{O_\Cat}=\under{O_\DCat}=\{0,1\}$ from Asssumption \ref{A:basic}(v), and let $\one_\Cat \xra{1_{O_\Cat}} O_\Cat$ and $\one_\DCat \xra{1_{O_\Cat}} O_\DCat$ denote the morphisms choosing the element $1$. Note that $\widehat{1_{O_\Cat}} = 1_{O_\DCat}$ by Remark \ref{rem:caniso}.

(b) Recall from Proposition \ref{E:reg} that a state $q: \one_\Cat \ra Q$ of a finite $T_\Sigma$-coalgebra $(Q,\gamma)$ accepts the language
\[ L_Q(q)=\{w\in\Sigma^* : \gamma_{\out}\cdot \gamma_w\cdot q = 1_{O_\Cat} \} \]
that we identify with the corresponding morphism of $\DCat$
\[L_Q(q): \Psi\Sigma^*\ra O_\DCat,\]
 see Notation \ref{N:OD}. Dually, for any $L_\Sigma$-algebra $(A,\alpha)$ equipped with a morphism $\alpha_{\out}: A\ra O_\DCat$ (representing a choice of final states), we define the \emph{language accepted by $\alpha_{\out}$} by
\[ L_A(\alpha_\out) = \{ w\in\Sigma^*: \alpha_{\out}\cdot \alpha_w \cdot \alpha_{\init} = 1_{O_\DCat} \}. \]
Using the unique ${L_{\Sigma}}$-algebra
homomorphism $e_A:\Psi{\Sigma}^*\to A$ of Remark \ref{rem:assalg}, this language corresponds to  the morphism of $\DCat$
\[ L_A(\alpha_{\out}) = {\alpha_{\out}}\cdot e_A:\Psi{\Sigma}^*\to{O_{\mathscr
D}}.\] 
\end{XNotation}

\begin{XDefinition}\ignorespaces\label{D:rev} The map $\Sigma^*\ra \Sigma^*$  reversing words extends uniquely to a morphism of $\DCat$
\[\rev_\Sigma: \Psi\Sigma^*\ra\Psi\Sigma^*.\]
The \emph{reversal}
of a language $L:\Psi{\Sigma}^*\to O_\DCat$ is the language
$L\cdot \rev_\Sigma$. 
\end{XDefinition}
Observe that $\rev_\Sigma$ is a $\DCat$-monoid morphism
\[ \rev_\Sigma: \Psi\Sigma^* \ra (\Psi\Sigma^*)^{op},\]
 where $(\Psi\Sigma^*)^{op}$ is the reversed monoid of $\Psi\Sigma^*$ with multiplication $x\bullet^{op} y = y\bullet x$. 

\begin{XLemma}\ignorespaces\label{L:rev} Let $(Q,\gamma)$ be a finite $T_\Sigma$-coalgebra and $(\widehat{Q},\widehat{\gamma})$ its dual $L_\Sigma$-algebra. Then the language accepted by a state $q: \one_\Cat \ra Q$ is the reversal of the
language accepted by $\widehat q:\widehat Q\to{O_{\mathscr D}}$:
\[L_{\widehat Q}(\widehat q)=L_Q(q)\cdot{\rev}_{\Sigma}.\]
\end{XLemma}

\begin{proof} 
A state $q: \one_\Cat\ra Q$ accepts a word $w=a_1\cdots a_n$ iff 
\[\gamma_{\out}\cdot \gamma_{a_n}\cdot \ldots \cdot \gamma_{a_1} \cdot q =  1_{O_\Cat}.\] This is dual to the equation
\[ \widehat{q} \cdot \widehat{\gamma_{a_1}} \cdot \ldots \cdot \widehat{\gamma_{a_n}} \cdot \widehat{\gamma_{\out}} = \widehat{1_{O_\Cat}} = 1_{O_\DCat},\]
which states precisely that $\widehat{q}$ accepts the word $w^{rev} = a_n\cdots a_1$ in the $L_\Sigma$-algebra $\widehat{Q}$. It follows that the two morphisms $L_{\widehat Q}(\widehat q)$ and $L_Q(q)\cdot{\rev}_{\Sigma}$ agree on $\Sigma^*$, hence they are equal.
\end{proof}

 One of the cornerstones of our Generalized Eilenberg Theorem in Section~\ref{S:GET} is a coalgebraic characterization of the closure under preimages, see Definition \ref{D:var}.  To this end we introduce first the
preimage $A^f$ of an ${L_{\Sigma}}$-algebra~$A$, and then the preimage $Q^f$ of a locally finite $T_\Sigma$-coalgebra $Q$. 

\begin{XNotation}\label{not:ax}
(a) Let $A$ be an object of $\DCat$. The object $[A,A]$ of endomorphisms, see Assumption \ref{A:basic}(iv), forms a $\DCat$-monoid with multiplication
 \[[A,A]\times[A,A] \ra [A,A],\quad (f,g) \mapsto g\cdot f,\]
 given by functional composition and unit $\id_A$. 
 
 (b) For an $L_\Sigma$-algebra $(A,\alpha)$ the
notation $\alpha_w:A\to A$ for words $w\in{\Sigma}^*$ (see Remark \ref{rem:assalg}) is
extended to $\alpha_x:A\to A$ for all $x\in\mathopen|\Psi{
\Sigma}^*\mathclose|$ as follows: since $\Psi\Sigma^*$ is the free $\DCat$-monoid on $\Sigma$ (see Proposition \ref{P:free}), the function
\[{\Sigma} \ra {\mathscr D}(A,A),\quad a \mapsto \alpha_a,\]
extends to a unique ${\mathscr D}$-monoid morphism  $\Psi{\Sigma}^* \ra [A,A]$ that we
denote by $x\mapsto\alpha_x$. \end{XNotation}

\begin{XDefinition}\label{def:af} Let
$f:\Psi{\Delta}^*\to\Psi{\Sigma}^*$ be a ${\mathscr
D}$-monoid morphism. For every ${L_{\Sigma}}$-algebra
${(A,\alpha)}$ we define its \emph{preimage under $f$} as the
${L_{\Delta}}$-algebra ${(A,\alpha)}^f=(A,\alpha^f)$ on the
same states $A$, with the same initial state, $\alpha^f_{
\init}={\alpha_{\init}}$, and with transitions
$\alpha^f_b=\alpha_{f(b)}$ for all $b\in{\Delta}$.
\end{XDefinition}\begin{XExample}\label{ex:af}
(a)~If ${\mathscr D}={\Set}$ or
${\Poset}$, we are given a monoid morphism
$f:{\Delta}^*\to{\Sigma}^*$. Every
${L_{\Sigma}}$-algebra $A$ yields an
${L_{\Delta}}$-algebra with transitions
$\alpha^f_b=\alpha_{f(b)}={\alpha_{a_n}}\cdot \ldots\cdot {\alpha_{a_1}}$ for
$f(b)=a_{1}\cdots a_{n}$. \endgraf(b)\enspace If ${\mathscr
D}=\JSL$, we
are given a semiring morphism $f:{{\mathcal
P}_{ f}{\Delta}^*}\to{{\mathcal P}_{
f}{\Sigma}^*}$. If the value $f(b)$ is a single word,
$f(b)=\{w\}$, then the corresponding transition is again
$\alpha^f_b=\alpha_w$. In general $f(b)=\{w_{1},\ldots,w_{k}\}$,
and since $\alpha_{({\mskip1.5mu \char"7B\mskip1.5mu })}$ is a
semilattice homomorphism, we conclude that
$\alpha^f_b={\alpha_{w_1}}\vee\cdots\vee{\alpha_{w_k}}$ (the
join in $[A,A]$). \endgraf(c)\enspace Analogously for ${\mathscr
D}=\Vect{\Int}$: if
$f(b)=\{w_{1},\ldots,w_{k}\}$, then $\alpha^f_b={\alpha_{w_1}}
\oplus\cdots\oplus{\alpha_{w_k}}$ (vector addition in $[A,A]$).

(d)~If ${\mathscr D}={\PSet}$ a zero-preserving monoid morphism
$f:{\Delta}^*+\{0\}\to{\Sigma}^*+\{0\}$ is given. The map $\alpha^f_b$ is defined as in (a) if $f(b)\neq 0$, and otherwise $\alpha^f_b(x) = \star_A$ for all $x\in\under{A}$, where $\star_A$ is the point of $A\in\PSet$.
\end{XExample}

\begin{XLemma}\label{lem:pre} Let $f:\Psi\Delta^*\ra \Psi\Sigma^*$ be a $\DCat$-monoid morphism. 

(a) $f$ is also an
${L_{\Delta}}$-algebra homomorphism
$f:\Psi{\Delta}^*\to(\Psi{\Sigma}^*)^f$.

(b) Every $L_\Sigma$-algebra homomorphism $h: A\ra A'$ is also an $L_\Delta$-algebra homomorphism $h: A^f\ra (A')^f$.

(c)\enspace Given an ${L_{\Sigma}}$-algebra $A$ we
have, in Remark~\ref{rem:assalg}, 
\[e_{A^f}=e_A\cdot
f:\Psi{\Delta}^*\to A^f.\]
\end{XLemma}

\begin{proof}
(a) Since $f(\epsilon)=\epsilon$ the initial state of $\Psi\Delta^*$ is mapped to the one of $(\Psi\Sigma^*)^f$. For any $a\in\Delta$, the $a$-transitions in $\Psi\Delta^*$ and $(\Psi\Sigma^*)^f$ are $\mathord{-}\bullet a$ and $\mathord{-}\bullet f(a)$, respectively. Hence preservation of transitions amounts to the equation
$f(x\bullet a)=f(x)\bullet f(a)$ for all $x\in\Psi\Delta^*$, which holds because $f$ is a $\DCat$-monoid morphism.

(b)~We clearly have $h \cdot \alpha_\init^f = h \cdot \alpha_\init = \alpha_\init' = (\alpha')_\init^f$. From $h\cdot \alpha_w = \alpha_w' \cdot h$ for all $w\in\Sigma^*$ we can conclude $h\cdot \alpha_x = \alpha_x' \cdot h$ for all $x\in\under{\Psi\Sigma^*}$. Indeed, both sides define $\DCat$-morphisms $\Psi\Sigma^* \ra [A,A']$ in the variable $x$ which agree on $\Sigma^*$. Thus they are equal. In particular, we have the desired equation $h\cdot \alpha_{f(a)} = \alpha_{f(a)}' \cdot h$
for all $a\in\Delta$.

(c) $\Psi{\Delta}^*$ is the initial $L_\Delta$-algebra, and by (a) and (b) both sides are
${L_{\Delta}}$-algebra homomorphisms.
\end{proof}

\begin{XCorollary}\ignorespaces\label{C:pre} Let $A$ be an
${L_{\Sigma}}$-algebra and ${\alpha_{
\out}}:A\to{O_{\mathscr D}}$ an output morphism. Then
${\alpha_\out}$ accepts in $A^f$ the preimage of the
language it accepts in~$A$: \[ L_{A^f}({\alpha_{\out}})=L_A({\alpha_{\out}})\cdot f:\Psi{\Delta}^*\to{O_{
\mathscr D}}.\] \end{XCorollary}

\begin{proof} 
Both sides are equal to ${\alpha_{\out}}\cdot e_A\cdot f$.
\end{proof}

\begin{XNotation}For any ${\mathscr
D}$-monoid morphism $f:\Psi{\Delta}^*\to\Psi{
\Sigma}^*$ we denote by $f^\dag$ the ${\mathscr D}$-monoid
morphism \[ \Psi\Delta^* \xra{\rev_\Delta} (\Psi\Delta^*)^{op} \xra{f} (\Psi\Sigma^*)^{op} \xra{\rev_\Sigma} \Psi\Sigma^* .\] 

 \end{XNotation}
\begin{XDefinition}\label{def:qf}Let
$f:\Psi{\Delta}^*\to\Psi{\Sigma}^*$ be a ${\mathscr
D}$-monoid morphism. For every finite
$T_\Sigma$-coalgebra ${(Q,\gamma)}$ we define its
\emph{preimage under $f$} as the $T_\Delta$-coalgebra
${(Q,\gamma)}^f=(Q,\gamma^f)$ whose dual is the preimage of the
dual $L_\Sigma$-algebra $(\widehat Q,\widehat\gamma)$ under $f^\dag$. Shortly: \[
\widehat{Q^f} = \widehat{Q}^{f^\dag}.\]
\end{XDefinition}

If $f=\Psi h^*$ for a function $h:\Delta\ra \Sigma$, it is easy to see that the coalgebra $Q^f$ of the previous definition coincides with the coalgebra $Q^h$ introduced for the definition of $\rho T$ (see \ref{def:ratfunc}). 

\begin{XExample} The preimage of $(Q,\gamma)$ under $f:\Psi{\Delta}^*\to\Psi{\Sigma}^*$ has the
same states and final states, and the transitions are given as follows:
\endgraf(a)\enspace Let ${\mathscr C}={\BA}$, ${\DL}$
and $f:\Delta^*\to \Sigma^*$. Letting $\alpha_a=\widehat{\gamma_a}$ we get, by Example \ref{ex:af}(a), the formula $\alpha^{f^\dag}_b = \alpha_{a_1}\cdot\ldots \cdot\alpha_{a_n}$ where $f(b)=a_1\cdots a_n$ (i.e., $f^\dag(b) = a_n\cdots a_1$). Since $\widehat{\gamma_b^f}=\alpha_b^{f^\dag}$ it follows that $\gamma_b^f=\gamma_{a_n}\cdot \ldots\cdot\gamma_{a_1}=\gamma_{f(b)}$.
\endgraf(b)\enspace  Let $\Cat = \JSL$ and $f: \Pow_f\Delta^* \to \Pow_f\Sigma^*$. We claim that $\gamma^f_b =
\gamma_{w_1}  \vee \cdots \vee \gamma_{w_k}$ where  $f(b)=\{w_1,\ldots,w_n\}$
 and the join is taken in $[Q,Q]$ (i.e., pointwise). Indeed, observe that the map $h\mapsto \widehat{h}$ gives a semilattice isomorphism $[Q,Q]\cong [\widehat{Q},\widehat{Q}]$. Letting $\alpha_a=\widehat{\gamma_a}$, this isomorphism maps $\gamma_{w_i}$ to $\alpha_{w_i^{rev}}$, and hence $\gamma_{w_1} \vee \cdots \vee \gamma_{w_k}$ to  $\alpha_{w_1^{rev}}\vee \cdots \vee \alpha_{w_n^{rev}}$. By Example \ref{ex:af}(b) this morphism is $\alpha_b^{f^\dag}$ since $f^\dag(b) = \{w_1^{rev},\ldots,w_n^{rev}\}$.
\endgraf(c)\enspace If $\Cat = \Vect{\Int}$ and $f:
\Pow_f\Delta^* \to \Pow_f\Sigma^*$, then $\gamma^f_b =
\gamma_{w_1} \oplus \cdots \oplus \gamma_{w_k}$ where $f(b) = \{w_1, \ldots,
w_k\}$. Indeed, the map $h\mapsto \widehat{h}$ gives an isomorphism of vector spaces $[Q,Q]\cong [\widehat{Q},\widehat{Q}]$. Letting $\alpha_a=\widehat{\gamma_a}$ this isomorphism maps $\gamma_{w_i}$ to $\alpha_{w_i^{rev}}$, and hence $\gamma_{w_1} \oplus \cdots \oplus \gamma_{w_k}$ to  $\alpha_{w_1^{rev}}\oplus \cdots \oplus \alpha_{w_n^{rev}}$. By Example \ref{ex:af}(c) this morphism is $\alpha_b^{f^\dag}$.
\endgraf(d)\enspace If ${\mathscr C}={\BR}$
and $f:\Delta^*+\{0\}\to \Sigma^*+\{0\}$, then $\gamma_b^f=\gamma_{f(b)}$ if $f(b)\neq 0$, and otherwise $\gamma_b^f$ is the zero map. The argument is similar to (a).
\end{XExample}

\begin{XProposition}\label{prop:pre}
The language accepted by a state $q$ of the
coalgebra $Q^f$ is the preimage under $f$ of the language $q$
accepts in~$Q$:
\[L_{Q^f}(q)=L_Q(q)\cdot f.\]
\end{XProposition}

\begin{proof} 
This follows from the computation
\[
\begin{array}{ll@{\qquad\quad}ll}
L_{Q^f}(q) &= L_{\widehat{Q^f}}(\widehat{q})\cdot \rev_\Delta & \text{(Lemma \ref{L:rev})}\\
&= L_{\widehat{Q}^{f^\dag}}(\widehat{q}) \cdot \rev_\Delta & \text{(def. $Q^f$)}\\
&= L_{\widehat Q}(\widehat{q}) \cdot f^\dag \cdot \rev_\Delta & \text{(Corollary \ref{C:pre})}\\
&\multicolumn{2}{l}{= L_{\widehat Q}(\widehat{q}) \cdot \rev_\Sigma\cdot f \cdot \rev_\Delta\cdot \rev_\Delta \quad \text{(def. $f^\dag$)}}\\
&= L_{\widehat Q}(\widehat{q}) \cdot \rev_\Sigma\cdot f \\
&= L_Q(q) \cdot f& \text{(Lemma \ref{L:rev}).} & \qedhere
\end{array}
\]
\end{proof}

\begin{XExample}
If $Q$ is finite subcoalgebra of $\rho T_\Sigma$ then $Q^f$ is the  $T_\Delta$-coalgebra of all
 languages $L:\Psi{\Sigma}^*\to{O_{\mathscr D}}$ in $Q$ 
with transitions given by left derivatives $\gamma_a(L)=f(a)^{-1}L$ for $a\in\Delta$. Here we extend the
notation $w^{-1}L$ for left derivatives from words
$w\in{\Sigma}^*$ to all elements $x$ of
$\Psi{\Sigma}^*$ as follows: let
$l_x:\Psi{\Sigma}^*\to\Psi{\Sigma}^*$ be the left
translation, $l_x(y)=x\bullet y$, then the left derivative
$x^{-1}L$ of a language $L:\Psi{\Sigma}^*\to{O_{\mathscr
D}}$ is $L\cdot l_x$.
\end{XExample}
We now extend the preimage concept from finite coalgebras to
locally finite ones. A $T_\Sigma$-coalgebra $Q$ is locally finite iff it is the filtered colimit of the diagram of all its finite subcoalgebras $Q_i\monoto Q$ (whose connecting $T_\Sigma$-coalgebra homomorphisms $d_{i,j}: Q_i\ra Q_j$ are inclusion maps). Given any ${\mathscr
D}$-monoid morphism $f:\Psi{\Delta}^*\to\Psi{
\Sigma}^*$, every $d_{i,j}$ is also a $T_\Delta$-coalgebra homomorphism $d_{i,j}: Q_i^f \ra Q_j^f$ by the dual of Lemma \ref{lem:pre}(b). Hence the coalgebras $Q_i^f$ and homomorphisms $d_{i,j}$ form a filtered diagram in $\Coalg{T_\Delta}$.

\begin{XDefinition}\label{def:qf2}
For every ${\mathscr
D}$-monoid morphism $f:\Psi{\Delta}^*\to\Psi{
\Sigma}^*$ and every locally finite $T_\Sigma$-coalgebra
$Q$ we denote by $Q^f$ the filtered colimit of the diagram of all $Q_i^f$, where $Q_i$ ranges over all finite subcoalgebras of $Q$.
\end{XDefinition}

\begin{XExample}\label{E:pre}
If $Q$ is a finite subcoalgebra of $\rho T_\Sigma$ then the languages accepted by $Q^f$ are precisely the languages $L\cdot f$ with $L\in \under{Q}$. This follows from the Proposition \ref{prop:pre} and the fact that every state $L$ of $Q$ accepts precisely the language $L$. Since $\rho T_\Sigma$ is the filtered colimit of its finite subcoalgebras $Q$, an analogous description holds for $(\rho T_\Sigma)^f$. Hence the unique $T_\Delta$-coalgebra homomorphism  $h: (\rho T_\Sigma)^f \ra \rho T_\Delta$ maps every language $L$ in $\rho T_\Sigma$ to its preimage $L\cdot f$.
\end{XExample}

Recall from Definition~\ref{D:var} the concept of closure under
preimages. This can now be formulated coalgebraically, much in the spirit of Proposition \ref{P:right}.

\begin{XTheorem}\ignorespaces\label{T:pre} A
subfunctor $V$ of the rational functor $\rho T$ is closed under
preimages iff for every ${\mathscr D}$-monoid morphism
$f:\Psi{\Delta}^*\to\Psi{\Sigma}^*$ there exists a
$T_\Delta$-coalgebra homomorphism from $(V{\Sigma})^f$
to $V{\Delta}$. \end{XTheorem}

\begin{proof} 
Suppose that
$k:(V{\Sigma})^f\to V{\Delta}$ is a
$T_\Delta$-coalgebra homomorphism. Composed with the
inclusion $i:V{\Delta}\hookrightarrow\rho{T_{\Delta}}$
it yields the homomorphism $h$ of Example~\ref{E:pre}
restricted to $(V{\Sigma})^f$ -- this follows from
$\rho{T_{\Delta}}$ being the terminal locally finite
$T_{\Delta}$-coalgebra. Thus $i\cdot k$ takes every language $L$ of
$\mathopen|V{\Sigma}\mathclose|$ to $L\cdot f$, proving
that $L\cdot f$ lies in $V{\Delta}$. 

For the converse, suppose that $V$ is closed under preimages. Then the
homomorphism $h$ of Example~\ref{E:pre} has a restriction $h_0:(V{\Sigma})^f\to
V{\Delta}$. That $h_0$ is a coalgebra homomorphism is a consequence of the following homomorphism theorem for coalgebras:  if $g:Q\to R$ and $i:R'\monoto R$
are $T_\Sigma$-coalgebra homomorphisms such that $i$ is injective and $g = i \cdot k$ for some morphism $k$ in ${\mathscr C}$, then
 $k$ is a coalgebra homomorphism. This follows easily from the observation that $T_\Sigma$ preserves monomorphisms.
\end{proof}

\section{Generalized Eilenberg Theorem}
\label{S:GET}

In this section we present our main result, the Generalized Eilenberg Theorem.
First we consider two ``finite'' versions of this theorem, proved by Kl\'ima and Pol\'{a}k \cite{KP} for the cases $\Cat = \BA$ and $\Cat = \DL$.
\begin{XDefinition}
A variety $V$ of
languages in ${\mathscr C}$ is
called \emph{object-finite} if $V{\Sigma}$ is finite for every alphabet ${\Sigma}$.
\end{XDefinition}

In the next theorem we will consider locally finite varieties of $\DCat$-monoids, see Definition \ref{def:locfin}. Recall that, in comparison to the \emph{pseudovarieties} of Definition \ref{def:pseudovar}, varieties of $\DCat$-monoids may contain infinite monoids and are closed under finite and infinite products. All locally finite varieties of $\DCat$-monoids form a lattice whose meet is intersection. The same holds for all object-finite varieties of languages where the intersection is taken objectwise.

\begin{XTheorem}[Generalized Eilenberg Theorem for Object\hyph Finite Varieties]\label{T:Eilenberg} 
The lattice of all object-finite varieties of languages in ${\mathscr C}$ is isomorphic to the lattice of all locally finite varieties of ${\mathscr D}$-monoids. 
\end{XTheorem} 
\begin{proof}[Proof sketch]
Given a variety $V$ of languages in $\Cat$, denote by $e_\Sigma: \Psi\Sigma^*\epito \widehat{V\Sigma}$ the dual finite $\Sigma$-generated monoid of the local variety $V\Sigma\monoto \rho T_\Sigma$, see Theorem \ref{prop:rqcmon}. Let $V^@$ be the class of all $\DCat$-monoids $D$ such that every $\DCat$-monoid morphism $h: \Psi\Sigma^*\ra D$, where $\Sigma$ is any finite alphabet, factorizes (necessarily uniquely) through $e_\Sigma$:
\[
h = (\xymatrix@1@+1pc{\Psi\Sigma^* \ar@{->>}[r]^-{e_\Sigma} & \widehat{V\Sigma} \ar@{-->}[r]^-{h'} & D}).
 \]
 A routine calculation shows that $V^@$ is a variety of $\DCat$-monoids whose free monoid on $\Sigma$ is $\widehat{V\Sigma}$. Since $\widehat{V\Sigma}$ is finite, $V^@$ is locally finite.
 
 Conversely, for a locally finite variety $W$ of $\DCat$-monoids (with free monoids $e_\Sigma: \Psi\Sigma^*\epito D_\Sigma$), we obtain an object-finite variety of languages $W^\square \monoto \rho T$ with $W^\square \Sigma$ defined by $\widehat{W^\square \Sigma} = D_\Sigma$ for all $\Sigma$.
 
 The constructions $V\mapsto V^@$ and  $W\mapsto W^\square$ are mutually inverse and hence define the desired lattice isomorphism.
 \end{proof}

\begin{XDefinition}An object-finite variety $V$ of languages in
${\mathscr C}$ is called \emph{simple} if it is generated by a
single alphabet ${\Sigma}$. That is, given any variety $V'$
such that $V{\Sigma}$ is a local subvariety of
$V'{\Sigma}$, then $V$ is a subfunctor of $V'$. A pseudovariety of $\DCat$-monoids is called \emph{simple} if it is is generated by a single finite $\DCat$-monoid $D$, i.e.,~all members of the pseudovariety are submonoids of quotients of
 finite powers $D^n$ ($n<\omega$).
\end{XDefinition}

\begin{XTheorem}[Generalized Eilenberg Theorem for Simple Varieties]\label{T:EilenbergSim} 
The poset of all simple varieties of languages in ${\mathscr C}$ is isomorphic to the poset of all simple pseudovarieties of ${\mathscr D}$-monoids. 
\end{XTheorem} 

\begin{proof}[Proof sketch]
For every locally finite variety $W$ of $\DCat$-monoids the class $W_f$ of finite members of $W$ forms a pseudovariety of $\DCat$-monoids. The isomorphism $V\mapsto V^@$  in the proof of Theorem \ref{T:Eilenberg}  restricts to one between simple varieties of languages and simple pseudovarieties of $\DCat$-monoids, that is, $V$ is simple iff $(V^@)_f$ is simple.
\end{proof}

Our main result now follows from Theorem~\ref{T:EilenbergSim} by a completion process. Recall that a \emph{\cpo} is a poset with directed joins. By a \emph{free \cpo-completion} of a poset $P^0$ is meant a \cpo $P$ containing $P^0$ as a subposet such that

(C1) every element $x$ of $P^0$ is compact
in $P$ (that is, whenever $x$ lies under a directed join $\bigvee p_i$ of elements of $P$, then $x\leq p_i$ for some $i$), and

(C2) the closure of $P^0$ under directed joins
is all of $P$. 

\noindent These two properties determine $P$ uniquely up to isomorphism. Concretely $P$ can be constructed as the set of all
ideals (= directed down-sets) of $P^0$, ordered by inclusion.

\begin{XTheorem}[Generalized Eilenberg Theorem]\label{T:Eil} 
The lattice of all varieties of languages in ${\mathscr C}$ is isomorphic to the lattice of all pseudovarieties of ${\mathscr D}$-monoids. 
\end{XTheorem}
\begin{proof}[Proof sketch]
One proves that

(1) the lattice $\mathscr L_{\mathscr C}$ of all varieties of languages in $\Cat$ is a free \cpo-completion of the poset $\mathscr L_{\mathscr C}^0$ of all simple varieties of languages, and

(2) the lattice $\mathscr L_{\mathscr D}$ of all pseudovarieties of $\DCat$-monoids is a free \cpo-completion of the poset $\mathscr L_{\mathscr D}^0$ of all simple pseudovarieties.

This requires a verification of the properties (C1) and (C2) above. Since $\mathscr L_{\mathscr C}^0\cong\mathscr L_{\mathscr D}^0$ by Theorem \ref{T:EilenbergSim}, and free \cpo-completions are unique up to isomorphism, it follows that $\mathscr L_{\mathscr C}\cong \mathscr L_{\mathscr D}$.
\end{proof}

For our five predualities of Example \ref{E:bool} we thus obtain the concrete correspondences in the table below as special cases of the Generalized Eilenberg Theorem. The second column describes the $\Cat$-algebraic operations under which varieties of languages are closed (in addition to closure under derivatives and preimages), and the fourth column characterizes the $\DCat$-monoids.\vspace{0.2cm}
\smallskip
\centerline{\vbox{\skip0=\baselineskip\skip1=\lineskip
\offinterlineskip\setbox0\hbox{boolean rings}
\setbox1\hbox{var.\ of languages} \setbox3\hbox{$\Int$-algebras} \halign{\vrule#&\kern.25em$\vcenter{\hsize=\wd0
\dimen0=\parindent\rightskip=0pt plus\dimen0 \leftskip=0pt
plus\dimen0 \spaceskip=\fontdimen2\font minus
\fontdimen4\font\setbox2\hbox{#\strut} \hyphenpenalty=0
\parindent=0pt \parfillskip=0pt \hbox to\wd0
{}{\lineskiplimit=0pt \baselineskip=.875\skip0 \ifdim\wd2>\wd0
\unhbox2 \else\hfill\vbox to 0pt{\vss\box2}\hfill\hbox{\strut
}\fi\vskip.125\skip0}\hbox to\wd0 {}}$\kern.25em\vrule&\kern
.25em$\vcenter{\hsize=\wd1 \dimen0=\parindent\rightskip=0pt
plus\dimen0 \leftskip=0pt plus\dimen0
\spaceskip=\fontdimen2\font minus \fontdimen4\font\setbox2\hbox
{#\strut} \hyphenpenalty=0 \parindent=0pt \parfillskip=0pt \hbox
to\wd1 {}{\lineskiplimit=0pt \baselineskip=.875\skip0
\ifdim\wd2>\wd1 \unhbox2 \else\hfill\vbox to
0pt{\vss\box2}\hfill\hbox{\strut}\fi\vskip.125\skip0}\hbox
to\wd1 {}}$\kern.25em\vrule&\kern.25em$\vcenter{\hsize=\wd0
\dimen0=\parindent\rightskip=0pt plus\dimen0 \leftskip=0pt
plus\dimen0 \spaceskip=\fontdimen2\font minus
\fontdimen4\font\setbox2\hbox{#\strut} \hyphenpenalty=0
\parindent=0pt \parfillskip=0pt \hbox to\wd0
{}{\lineskiplimit=0pt \baselineskip=.875\skip0 \ifdim\wd2>\wd0
\unhbox2 \else\hfill\vbox to 0pt{\vss\box2}\hfill\hbox{\strut
}\fi\vskip.125\skip0}\hbox to\wd0 {}}$\kern.25em\vrule&\kern
.25em$\vcenter{\hsize=\wd3 \dimen0=\parindent\rightskip=0pt
plus\dimen0 \leftskip=0pt plus\dimen0
\spaceskip=\fontdimen2\font minus \fontdimen4\font\setbox2\hbox
{#\strut} \hyphenpenalty=0 \parindent=0pt \parfillskip=0pt \hbox
to\wd3 {}{\lineskiplimit=0pt \baselineskip=.875\skip0
\ifdim\wd2>\wd3 \unhbox2 \else\hfill\vbox to
0pt{\vss\box2}\hfill\hbox{\strut}\fi\vskip.125\skip0}\hbox
to\wd3 {}}$\kern.25em\vrule\cr\noalign{\hrule} &\textsc{category
${\mathscr C}$}&var.\ of languages are closed under&\textsc{category
${\mathscr D}$}&${\mathscr D}$-monoids\cr\noalign{\hrule}
&$\BA$&$\cup$, $\cap$, $\overline{({\mskip1.5mu
\char"7B\mskip1.5mu })}$, $\emptyset$,
${\Sigma}^*$&$\Set$&monoids\cr&$\DL$&$\cup$, $\cap$, $\emptyset$, ${\Sigma}^*$&$\Poset$&ordered monoids\cr&$\JSL$&$\cup$, $\emptyset$&$\JSL$&idempotent semirings\cr&$\Vect{\Int}$&$\oplus$, $\emptyset$&$\Vect{\Int}$&$\Int$-algebras\cr&$\BR$&$\oplus$, $\cap$, $\emptyset$&$\PSet$& monoids with $0$\cr\noalign{\hrule}}}}

The cases $\Cat = \BA$, $\DL$ and $\Vect{\Int}$ are due to Eilenberg \cite{E}, Pin \cite{Pin95} and Reutenauer \cite{R}, respectively. The case $\Cat = \JSL$ is ``almost'' the result of Pol\'ak \cite{P}: his \emph{disjunctive varieties of languages} are required to contain $\Sigma^*$ for every $\Sigma$, and he considers semirings without $0$. In our setting this would mean to take the predual categories $\Cat=\mathbf{JSL_{01}}$ (join-semilattices with $0$ and $1$) and $\DCat=\mathbf{JSL}$ (join-semilattices). We opted for the more symmetric preduality $\Cat=\DCat=\JSL$ as semirings are usually considered with a zero element.
The last example, $\Cat=\BR$, is a new variant of Eilenberg's theorem.

\section{Varieties of Regular Behaviors}\label{sec:regbeh}

Although all our results so far concerned acceptors and
varieties of languages they accept, we can with little effort
generalize the whole theory to Moore automata, where the output
morphism ${\gamma_{\out}}$ has, in lieu of ${\{0,1\}}$, any
finite set $O$ (of outputs) as codomain. The role of languages
over ${\Sigma}$ is now taken over by functions
$\beta:{\Sigma}^*\to O$, and the role of regular languages by \emph{regular behaviors}, i.e,
those $\beta$ realized by a state of some finite Moore automaton.

Given a fixed finite set $O$ of outputs, all we need to change in the previous text is
Assumption~\ref{A:basic}(v) which is replaced by $\mathopen|{O_{\mathscr C}}\mathclose|=O$, where $O_\Cat$ is the object dual to $\one_\DCat$. For the object $O_\DCat$ dual to $\one_\Cat$ we can assume, as in Remark \ref{rem:caniso}, that $\under{O_\Cat}=\under{O_\DCat}$. Moore automata are modeled as coalgebras for the endofunctor on $\Cat$
\[ T_\Sigma Q = O_\Cat \times Q^\Sigma.\]
\begin{XExample}[Linear weighted automata] Let
${\mathscr C}={\mathscr D}$ be the category
$\Vect{K}$ of vector spaces over a finite
field $\mathbf{K}$. Here ${O_{\mathscr C}}=\mathbf{K}$, the one-dimensional space.
Thus, a $T_\Sigma$-coalgebra is a \emph{linear weighted automaton}: it consists of a vector space
$Q$ of states, a linear output function ${\gamma_{\out}}:Q\to \mathbf{K}$ and
linear transitions $\gamma_a:Q\to Q$ for $a\in\Sigma$. \end{XExample}

The rational fixpoint $\rho T_\Sigma$ of the functor $T_\Sigma$ is carried by the set of all regular behaviors $\beta:{\Sigma}^*\to O$ (that we identify with the corresponding morphisms $\beta:\Psi\Sigma^*\ra O_\DCat$ in $\DCat$). Its output map assigns the value
at the empty word, ${\gamma_{\out}}(\beta)=\beta(\varepsilon)
$, and its transitions are given by \emph{left derivatives} $\gamma_a(\beta)=\beta(a\cdot{
\mskip1.5mu \char"7B\mskip1.5mu })$ for all $a\in{\Sigma}$.
Symmetrically, the \emph{right derivatives} of $\beta$ are the functions $\beta(\mathord{-}\cdot a)$ for $a\in \Sigma$. We define the \emph{rational functor} $\rho T:{\Set}_{ f}^{op}\to{\mathscr C}$ in complete analogy to Definition \ref{def:ratfunc}.

\begin{XDefinition} A subfunctor
$V\hookto \rho T$ is \emph{closed
under preimages} if for every ${\mathscr D}$-monoid morphism
$f:\Psi{\Delta}^*\to\Psi{\Sigma}^*$ and every behavior
$\beta:\Psi{\Sigma}^*\to{O_{\mathscr D}}$ in
$V{\Sigma}$ the behavior $\beta\cdot f$ lies in
$V{\Delta}$. A \emph{variety of behaviors in
${\mathscr C}$} is a subfunctor $V\hookto \rho T$ closed
under preimages and (left and right) derivatives. 
\end{XDefinition}
\begin{XExample}[Varieties of rational power series]
If $\Cat=\DCat=\Vect{K}$, the coalgebra $\rho T_\Sigma$
consists of all rational power series $\beta:\Sigma^*\ra \mathbf{K}$, i.e., behaviors of linear weighted automata. Moreover, $\DCat$-monoids are precisely algebras over the field $\mathbf{K}$, and the free $\mathbf{K}$-algebra $\Psi\Sigma^*$ is carried by the set of all functions $\Sigma^*\ra \mathbf{K}$ with finite support. Given a morphism $f: \Psi\Delta^*\ra\Psi \Sigma^*$ of free $\mathbf{K}$-algebras, the preimage of a power series $\beta:\Sigma^*\ra \mathbf{K}$ under $f$ is \[\beta': \Delta^*\ra \mathbf{K},\quad \beta'(w) = \sum_{v\in\Sigma^*} f(w)(v)\cdot \beta(v).\]
(This sum is well-defined because $f(w): \Sigma^*\ra \mathbf{K}$ has finite support.) A variety of behaviors in $\Vect{K}$ forms linear subspaces $V\Sigma$ of $\rho T_\Sigma$ closed under derivatives and preimages of $\mathbf{K}$-algebra morphisms.
This coincides with the concept of a \emph{variety of rational power series} introduced by Reutenauer \cite{R}. His Theorem III.1.1 is therefore a special
case of \end{XExample}

\begin{XTheorem}[Generalized Eilenberg Theorem for Regular Behaviors] 
The lattice of all varieties of behaviors in ${\mathscr C}$ is isomorphic to the lattice of all pseudovarieties of ${\mathscr D}$-monoids. 
\end{XTheorem}

\section{Conclusions and Future Work}\label{sec:conc}

In the present paper we demonstrated that Eilenberg's variety theorem, a central result of algebraic automata theory, holds at the level of an abstract duality between (algebraic) categories. Our result covers uniformly several known extensions and refinements of Eilenberg's theorem and also provides a new Eilenberg-type correspondence for pseudovarieties of monoids with $0$.

In the future we intend to classify all predual pairs $({\mathscr C},{\mathscr D})$ satisfying our Assumptions~\ref{A:basic} above. Using the \emph{natural duality} framework of Clark and Davey \cite{CD}, we expect to show that only finitely many Eilenberg theorems exist for every fixed finite output set $O$. 

We also believe that more can be said about the relationship between local varieties of languages and varieties of languages. By studying languages and monoids in a setting of Grothendieck fibrations, it seems possible that the global Eilenberg theorem turns out to be an instance of the local one -- or that both theorems are instances of one and the same result. This might lead to independent proofs of our results, and at the same time to a more abstract and hence illuminating view of the concepts involved.

Another important direction is the connection of regular languages to profinite algebras. All results mentioned in the Related Work can be interpreted in the setting $(\Cat,\DCat)=(\BA,\Set)$ or $(\DL,\Poset)$, and we aim to extend them to general pairs of predual categories.  This will require the generalization of two core concepts of algebraic automata theory, the syntactic (ordered) monoid associated to a regular language and the free profinite (ordered) monoid on an alphabet $\Sigma$, to a notion of \emph{syntactic $\DCat$-monoid} and \emph{free profinite $\DCat$-monoid}, respectively. We conjecture that the free profinite $\DCat$-monoid on $\Sigma$ arises as the limit of all quotient monoids of $\Psi\Sigma^*$, and hope to derive a generalized Reiterman theorem. In the classical setting the theorems of Eilenberg and Reiterman are two key ingredients in a great success of algebraic automata theory: these results allow to specify classes of regular languages (e.g.~the star-free languages) by profinite identities, which leads to decidability results for such classes. Perhaps, new such results are enabled through subsequent work in our generalized setting.

Finally, we are interested in extending our results from regular languages and Moore behaviors to other notions of rational behavior, such as $\omega$-regular languages or regular tree languages. Here the role of monoids is taken over by two-sorted algebras called \emph{Wilke algebras} (or \emph{right binoids}, see Wilke~\cite{W}) and \emph{forest algebras} (introduced by Bojanczyk and Walukiewicz \cite{BW}), respectively. The main challenge will be to identify the proper categorical model for the corresponding acceptors, B\"uchi automata and tree automata.

\clearpage
\appendices
\section{Proofs}

This appendix  provides all proofs we omitted due to space limitations, along with some technical lemmas required for these proofs.

\begin{XRemark}\label{rem:fpcoalg}
(a) Recall that an object $A$ of a category $\ACat$ is called \emph{finitely presentable} if its $\hbox{hom}$-functor
${\mathscr A}(A,{\mskip1.5mu \char"7B\mskip1.5mu }): \ACat\ra\Set$ preserves
filtered colimits. 
If $\ACat$ is a locally finite variety of (ordered) algebras, the finitely presentable objects are precisely the finite algebras.

(b) Every finite $T_\Sigma$-coalgebra is a finitely presentable object of $\Coalg{T_\Sigma}$, see \cite{AP}. Hence, given a filtered colimit cocone $c_i': Q_i'\ra Q'$ ($i\in I$) in $\Coalg{T_\Sigma}$, every coalgebra homomorphism $h: Q\ra Q'$ with finite domain $Q$ factorizes (in $\Coalg{T_\Sigma}$) through some $c_i'$:
\[  
\xymatrix{
Q \ar[r]^h \ar@{-->}[dr]_{h'} & Q'\\
& Q_i' \ar[u]_{c_i'}
}
\]
\end{XRemark}

\begin{XRemark}\label{rem:creatcol}
The forgetful functor $\Coalg{T_\Sigma}\ra\Cat$ preserves and creates colimits. The latter means that, given a diagram $(Q_i,\gamma_i)$ ($i\in I$) of $T_\Sigma$-coalgebras and a colimit cocone $(c_i: Q_i\ra Q)_{i\in I}$ in $\Cat$, there is a unique $T_\Sigma$-coalgebra structure $\gamma$ on $Q$ for which the maps $c_i$ are $T_\Sigma$-coalgebra homomorphisms $c_i: (Q_i,\gamma_i)\ra (Q,\gamma)$. Moreover,  $(c_i)$ is a colimit cocone in $\Coalg{T_\Sigma}$. 

The uniqueness of $\gamma$ gives rise to a useful proof principle: if two coalgebra structures $\gamma$ and $\gamma'$ on $Q$ are given such that each $c_i$ is a coalgebra homomorphism $c_i: (Q_i,\gamma_i)\ra(Q,\gamma)$ and $c_i: (Q_i,\gamma_i)\ra(Q,\gamma')$, it follows that $\gamma=\gamma'$.
\end{XRemark}

The preimage construction  $Q^f$, see Definition \ref{def:qf2}, has the following equivalent formulation:

\begin{XConstruction}\label{cons:qf}
For every locally finite $T_\Sigma$-coalgebra $(Q,\gamma)$ and $\DCat$-monoid morphism $f:\Psi\Delta^*\ra\Psi\Sigma^*$ we construct a $T_\Delta$-coalgebra $(Q,\gamma)^f$ as follows:

(1) Express $(Q,\gamma)$ as a filtered colimit $c_i: (Q_i,\gamma_i)\ra (Q,\gamma)$ ($i\in I$) of finite $T_\Sigma$-coalgebras. 

(2) Let $\gamma^f$ be the unique $T_\Delta$-coalgebra structure on $Q$ for which all $c_i: (Q_i,\gamma_i)^f \ra (Q,\gamma^f)$
are $T_\Delta$-coalgebra homomorphism, see Remark \ref{rem:creatcol}. Put
\[ (Q,\gamma)^f := (Q,\gamma^f).\]
\end{XConstruction}

The following lemma summarizes some important properties of this construction.

\begin{XLemma}\label{lem:qfprops}
Let $(Q,\gamma)$ be a locally finite $T_\Sigma$-coalgebra and $f:\Psi\Delta^*\ra\Psi\Sigma^*$ a $\DCat$-monoid morphism.

(a) The coalgebra structure of $(Q,\gamma)^f$ is independent of the choice of the colimit cocone $(c_i)$ in Construction \ref{cons:qf}.

(b) If $f = \Psi f_0^*$ for some $f_0: \Delta\ra\Sigma$, then $(Q,\gamma)^f$ has the coalgebra structure
\begin{equation}\label{eq:qf0}
  Q \xra{\gamma}{O_{\mathscr C}}\times Q^{\Sigma} \xra{\id\times Q^{f_0}} {O_{\mathscr C}}\times Q^{\Delta},
\end{equation}
c.f. the definition of the rational functor $\rho T$ (Definition \ref{def:ratfunc}).
(c) Every homomorphism $h:Q\ra Q'$ of locally finite $T_\Sigma$-coalgebras is also a homomorphism $h: Q^f\ra (Q')^f$ of $T_\Delta$-coalgebras.

(d) For any $\DCat$-monoid morphism $g:\Psi\Gamma^*\ra \Psi\Delta^*$,
\[ (Q^f)^g = Q^{f\cdot g}.\]

(e) The construction $(\mathord{-})^f$ commutes with coproducts: given locally finite $T_\Sigma$-coalgebras $Q_j$ ($j\in J$), we have
\[ (\coprod_j Q_j)^f = \coprod_j Q_j^f. \]
\end{XLemma}

\begin{proof}
(a) Suppose $\gamma^f$ has been defined by means of the cocone $(c_i)$, and another filtered colimit cocone \[c_j': (Q_j',\gamma_j')\ra (Q,\gamma)\quad (j\in J)\] with $Q_j'$ finite is given. By Remark \ref{rem:creatcol} it suffices to show that the maps $c_j'$ are $T_\Delta$-coalgebra homomorphisms $c_j': (Q_j',\gamma_j')^f \ra (Q,\gamma^f)$. 

Given $j\in J$, there exists by Remark \ref{rem:fpcoalg} a $T_\Sigma$-coalgebra homomorphism $g: (Q_j',\gamma_j')\to (Q_i,\gamma_i)$ with $c_i\cdot g = c_j'$ for some $i$. It follows that $c_j'$ is a $T_\Delta$-coalgebra homomorphism, being the composite of the $T_\Delta$-coalgebra homomorphisms
\[ (Q_j',\gamma_j')^f \xra{g} (Q_i,\gamma_i)^f \xra{c_i} (Q,\gamma^f).\]
Indeed, $g$ is a $T_\Delta$-coalgebra homomorphism using the definition of $Q^f$ for finite $Q$ (Definition~\ref{def:qf}) and Lemma~\ref{lem:pre}(b), and $c_i$ is one by the definition of $\gamma^f$.

(b) Given an $L_\Sigma$-algebra $(A,\alpha)$ the $L_\Delta$-algebra $(A,\alpha)^f$ (see Definition \ref{def:af}) has the transitions $\alpha^f_a = \alpha_{f_0(a)}$ for $a\in\Delta$. 
Hence, for a finite $T_\Sigma$-coalgebra $(Q,\gamma)$, the $L_\Delta$-algebra $\widehat{Q^f} = \widehat{Q}^{f^\dag} = \widehat{Q}^f$ has transitions $\widehat{\gamma_{f_0(a)}}: \widehat{Q}\ra\widehat{Q}$ for $a\in \Delta$. Dually $Q^f$ has the transitions $\gamma_{f_0(a)}: Q\ra Q$, which are precisely the transitions corresponding to the coalgebra structure \eqref{eq:qf0}.

In the case where $(Q,\gamma)$ is just locally finite, express $(Q,\gamma)$ as a filtered colimit $c_i: (Q_i,\gamma_i)\ra (Q,\gamma)$ ($i\in I$) of finite $T_\Sigma$-coalgebras, and consider the diagram below:
\[ 
\xymatrix{
  Q \ar[r]^<<<<<{\gamma}  & {O_{\mathscr C}}\times Q^{\Sigma} \ar[r]^{\id\times Q^{f_0}} & {O_{\mathscr C}}\times Q^{\Delta}\\
 Q_i \ar[u]^{c_i} \ar[r]_<<<<<{\gamma_i}  & {O_{\mathscr C}}\times Q_i^{\Sigma} \ar[r]_{\id\times Q_i^{f_0}} \ar[u]_{T_\Sigma c_i} & {O_{\mathscr C}}\times Q_i^{\Delta} \ar[u]_{T_\Delta c_i} \\
}
\]
This diagram commutes because $c_i$ is a $T_\Sigma$-coalgebra homomorphism and by naturality. The lower row is the coalgebra structure of $(Q_i,\gamma_i)^f$ by the first part of the proof. Hence $c_i$ is a $T_\Delta$-coalgebra homomorphism from $(Q_i,\gamma_i)^f$ to the coalgebra in the upper row, which implies that the upper row defines the coalgebra structure of $(Q,\gamma)^f$.

(c) Express $Q$ and $Q'$ as filtered colimits $c_i: Q_i\ra Q$ ($i\in I$) and $c_j': Q_j'\ra Q'$ ($j\in J$) of finite $T_\Sigma$-coalgebras. By Remark \ref{rem:fpcoalg} there exists for every $i\in I$ some $T_\Sigma$-coalgebra homomorphism $g: Q_i\ra Q_j'$ for which the diagram below commutes:
\[
\xymatrix{
Q \ar[r]^h & Q'\\
Q_i \ar[u]^{c_i} \ar[r]_g & Q_j' \ar[u]_{c_j'}
}
\]
It follows that $h\cdot c_i: Q_i^f\ra (Q')^f$ is a $T_\Delta$-coalgebra homomorphism, being the composite of the $T_\Delta$-coalgebra homomorphisms $Q_i^f \xra{g} (Q_j')^f \xra{c_j'} (Q')^f$ (one uses the same argument as in point~(a)). Since the morphisms $c_i$ are jointly epimorphic in $\Cat$, it follows that $h$ is a $T_\Delta$-coalgebra homomorphism $h: Q^f\ra (Q')^f$.

(d) (i) We first show that for all $L_\Sigma$-algebras $(A,\alpha)$ we have
\[ (A^f)^g = A^{f\cdot g}. \]
Indeed, both algebras have states $A$ and initial states $\alpha_\init$. To see that they have the same transitions, consider the diagram below, cf. Notation \ref{not:ax} and Definition \ref{def:af}. 
\begin{equation}
  \label{eq:diagabove}
  \vcenter{
    \xymatrix@C+1pc{
      \Psi\Sigma^* \ar[r]^-{x\mapsto \alpha_x} & [A,A]\\
      \Psi\Gamma^* \ar[u]^{f\cdot g} \ar[r]_g & \Psi\Delta^* \ar[u]_{y\mapsto \alpha_y^f} \ar[ul]|{f}
    }
  }
\end{equation}
The upper triangle commutes by the definition of $\alpha^f$. Hence, for all $a\in\Gamma$,
\begin{align*}
(\alpha^{f\cdot g})_a &= \alpha_{fg(a)} & \text{(def. $\alpha^{f\cdot g}$)}\\
&= (\alpha^f)_{ga} & \text{(diagram~\refeq{eq:diagabove})}\\
&= ((\alpha^f)^g)_a & \text{def. $(\alpha^f)^g$}
\end{align*}

(ii) If $Q$ is a finite $T_\Sigma$-coalgebra we conclude from (i):
\[ \widehat{Q^{f\cdot g}} = \widehat{Q}^{(f\cdot g)^\dag} = \widehat{Q}^{f^\dag\cdot g^\dag} = (\widehat{Q}^{f^\dag})^{g^\dag} = \widehat{(Q^f)^g},\]
so $Q^{f\cdot g}= (Q^f)^g$.

(iii) Now let $Q$ be locally finite, and express $Q$ as a  filtered colimit $c_i: Q_i\ra Q$ ($i\in I$) of finite $T_\Sigma$-coalgebras. Hence by (c) we have $T_\Gamma$-coalgebra homomorphisms $c_i: Q_i^{f\cdot g} \ra Q^{f\cdot g}$ and $c_i: (Q_i^f)^g \ra (Q^f)^g$, where $Q_i^{f\cdot g}=(Q_i^f)^g$ by (ii) above. It follows that $Q^{f\cdot g}$ and $(Q^f)^g$ have the same coalgebra structure.

(e) (i) We first prove that, for each family of $L_\Sigma$-algebras $(A_j,\alpha_j)$ ($j\in J$),
\[ (\prod_j A_j)^f = (\prod_j A_j^f).\]
Clearly the algebras on both sides of the equation have the same states $\prod_j A_j$ and the same initial state. Concerning the transitions, consider the commutative diagram below:
\[
\xymatrix{
\Psi\Delta^* \ar[r]^f & \Psi\Sigma^* \ar[rr]^-{a\mapsto \prod (\alpha_j)_a} \ar[drr]_(.4)*+{\labelstyle\langle a\mapsto (\alpha_j)_a \rangle} && [\prod A_j,\prod A_j] \\
&&& \prod [A_j,A_j] \ar[u]_{(f_j) \mapsto \prod f_j}
}
\]
The upper and lower path define the transitions of $(\prod_j A_j)^f$ and $\prod_j A_j^f$, respectively. Hence they have the same transitions.

(ii) Suppose now that $J$ is finite and finite $T_\Sigma$-coalgebras $Q_j$ are given. Then we conclude from (i) and duality:
\[ \widehat{(\coprod_j Q_j)^f} = (\widehat{\coprod Q_j})^{f^\dag} = (\prod \widehat{Q_j})^{f^\dag} = \prod \widehat{Q_j}^{f^\dag}  = \widehat{\coprod Q_j^f}.\]
The statement for arbitrary $J$ and locally finite coalgebras $Q_j$ now follows from the fact that filtered colimits and coproducts commute in $\Coalg{T_\Sigma}$, and every infinite coproduct is a filtered colimit of finite ones.
\end{proof}

\begin{XRemark}\label{rem:homcoalg}
Recall the \emph{homomorphism theorem} for coalgebras: if a $T_\Sigma$-coalgebra homomorphism $g:Q\to R$ factorizes in $\Cat$ through a subcoalgebra $i:R'\monoto R$, then
the factorizing morphism $g'$ is a coalgebra homomorphism. 
\[
\xymatrix@-1pc{
& R \\
Q \ar[ru]^g \ar@{-->}[rr]_{g'} && R' \ar@{ >->}[lu]_i
}
 \]
This follows easily from the observation that $T_\Sigma$ preserves monomorphisms.
\end{XRemark}

\begin{XRemark}\label{rem:sub}
  Note that locally finite $T_\Sigma$-coalgebras are closed under subcoalgebras. Indeed, suppose that $i: Q' \monoto Q$ is a $T_\Sigma$-coalgebra homomorphism, where $Q$ is locally finite. Now for every $q \in |Q'|$ we have a finite subcoalgebra $Q_0$ of $Q$ containing $q$. Since $T_\Sigma$ preserves intersections we can conclude that $Q_0 \cap Q'$ is a subcoalgebra of $Q'$ containing $q$. 
\end{XRemark}

\begin{XLemma}\label{lem:subfunc}
A family of subcoalgebras $m_\Sigma: V\Sigma \monoto \rho T_\Sigma$ ($\Sigma\in\Set_f$) forms a subfunctor of $\rho T$ iff, for every function $f_0:\Delta\to\Sigma$ in $\Set_f$, a $T_\Sigma$-coalgebra homomorphism from $(V\Sigma)^{f}$ to $V\Delta$ exists, where  $f=\Psi f_0^*$.
\end{XLemma}

\begin{proof}
Suppose that the  $V\Sigma$ form a subfunctor $V\monoto\rho T$, and let $f_0: \Delta\ra\Sigma$. Then we have the commutative diagram below with $g=Vf_0$:
\[
\xymatrix{
\rho T_\Sigma \ar[r]^{\rho T_{f_0}} & \rho T_\Delta\\
\ar@{ >->}[u]^{m_\Sigma} \ar[r]_{g} V\Sigma & V\Delta \ar@{ >->}[u]_{m_\Delta}
}
\]
By definition, $\rho T_{f_0}$ is a $T_\Delta$-coalgebra homomorphism $\rho T_{f_0}: (\rho T_\Sigma)^f \ra \rho T_\Delta$, and by Lemma~\ref{lem:qfprops}(b), $m_\Sigma$ is a $T_\Delta$-coalgebra homomorphism $m_\Sigma: (V\Sigma)^f \ra (\rho T_\Sigma)^f$. The homomorphism theorem then implies that $g$ is a $T_\Delta$-coalgebra homomorphism from $(V\Sigma)^f$ to $V\Delta$.

Conversely, given a $T_\Delta$-coalgebra homomorphism $g: (V\Sigma)^f\ra V\Delta$, the above square commutes because both $\rho T_{f_0}\cdot m_\Sigma$ and $m_\Delta\cdot g$ are $T_\Delta$-coalgebra homomorphisms from $(V\Sigma)^f$ into the terminal locally finite $T_\Delta$-coalgebra $\rho T_\Delta$. Note that $(V\Sigma)^f$ is locally finite because $V\Sigma$ is locally finite, being a subcoalgebra of the locally finite coalgebra $\rho T_\Sigma$, see Remark~\ref{rem:sub}.
\end{proof}

\begin{XRemark}\label{rem:homtheo}
Dually to Remark \ref{rem:homcoalg}, we have the following homomorphism theorem for $L_\Sigma$-algebras: if an $L_\Sigma$-algebra homomorphism $f:A\to C$ factorizes in $\DCat$ through an $L_\Sigma$-quotient algebra  $e:A\epito B$, then
the factorizing morphism $f'$ is an $L_\Sigma$-algebra homomorphism.
\[
\xymatrix@-1pc{
& A \ar@{->>}[dl]_e \ar[dr]^f &\\
B \ar@{-->}[rr]_{f'} && C
}
\]
The analogous statement holds for $\DCat$-monoid morphisms. Since $\Alg{L_\Sigma}$ and $\Mon{\DCat}$ are varieties of (ordered) algebras, this is a special case of the well-known homomorphism theorem of universal algebra.
\end{XRemark}

\begin{XLemma}\label{lem:commdiag}
Let $Q\monoto \rho T_\Sigma$ and $Q'\monoto \rho T_\Delta$ be finite local varieties of languages, and $f:\Psi\Delta^* \ra \Psi\Sigma^*$ a $\DCat$-monoid morphism. Then  there exists a $T_\Delta$-coalgebra homomorphism from $h: Q^f\ra Q'$ iff there exists a $\DCat$-monoid morphism $g$ making the following square commute:
\begin{equation}\label{diag:mon}
\begin{gathered}
\xymatrix{
\Psi\Delta^* \ar[r]^{f^\dag}  \ar@{->>}[d]_{e_{\widehat{Q'}}} & \Psi\Sigma^* \ar@{->>}[d]^{e_{\widehat Q}} \\
\widehat{Q'} \ar[r]_{g} & \widehat{Q}
}
\end{gathered}
\end{equation}
In this case, we have $g=\widehat h$.
\end{XLemma}

\begin{proof}
Given a $T_\Delta$-coalgebra homomorphism $h: Q^f \ra Q'$ we have, by Lemma \ref{lem:pre}, the following square of $L_\Delta$-algebra homomorphisms
\begin{equation}\label{diag:lalg}
\begin{gathered}
\xymatrix{
\Psi\Delta^* \ar[r]^{f^\dag}  \ar@{->>}[d]_{e_{\widehat{Q'}}} & (\Psi\Sigma^*)^{f^\dag} \ar@{->>}[d]^{e_{\widehat{Q}}} \\
\widehat{Q'} \ar[r]_{g} & \widehat{Q}^{f^\dag}
}
\end{gathered}
\end{equation}
where $g=\widehat{h}$. This diagram commutes because $\Psi\Delta^*$ is the initial $L_\Delta$-algebra. Therefore the square \eqref{diag:mon} commutes. That $g$ is a $\DCat$-monoid morphism follows from Remark \ref{rem:homtheo}.

Conversely, given a morphism $g$ for which \eqref{diag:mon} commutes, then $g: \widehat{Q'}\ra \widehat{Q}^{f^\dag} = \widehat{Q^f}$ is an $L_\Delta$-algebra homomorphism by \eqref{diag:lalg} and Remark \ref{rem:homtheo}, so dually a $T_\Delta$-coalgebra homomorphism $Q^f\ra Q'$ exists.
\end{proof}

\begin{XNotation}\label{not:vat}
Let $V$ be an object-finite variety of languages in $\Cat$. 

(a) We denote by $e_\Sigma: \Psi \Sigma^* \epito \widehat{V\Sigma}$ the $\Sigma$-generated $\DCat$-monoid  corresponding
to the local variety $V{\Sigma}$. 

(b) Since $V$ is closed under preimages, we have a (unique) $T_\Delta$-coalgebra homomorphism
$\ol Vf:(V{\Sigma})^f\to V{\Delta}$ for every $\DCat$-monoid morphism $f: \Psi\Delta^*\ra\Psi\Sigma^*$, see Theorem~\ref{T:pre}. 
This notation is in good harmony with the functor notation for $V$: for a function $f_0:{\Delta}\to{\Sigma}
$ and the corresponding $\DCat$-monoid morphism $f=\Psi
f_0^*:\Psi{\Delta}^*\to\Psi{\Sigma}^*$ we have $Vf_0= \ol Vf$: indeed, both maps are $T_\Delta$-coalgebra homomorphisms from $(V\Sigma)^f$ to $V\Delta$ and for the $T_\Delta$-coalgebra homomorphism $i: V\Delta \monoto \rho T_\Delta$ we have that $i \cdot Vf_0 = i \cdot \ol Vf : (V\Sigma)^f \to \rho T_\Delta$ agree by the finality of $\rho T_\Delta$; now use that $i$ is monomorphic to conclude $Vf_0 = \ol V f$.

(c)  We denote by $V^@$ the class of
all ${\mathscr D}$-monoids $D$ such that every $\DCat$-monoid
homomorphism  $h: \Psi{\Sigma}^*\to D$, where ${\Sigma}$ is
any finite set, factorizes (necessarily uniquely) through $e_{\Sigma}$.
\[
\xymatrix{
\Psi\Sigma^* \ar@{->>}[d]_{e_\Sigma} \ar[dr]^h &\\
\widehat{V\Sigma} \ar@{-->}[r] & D
}
\]
\end{XNotation}

\begin
{XProposition}\ignorespaces\label{P:fin} $V^@$ is a locally
finite variety of ${\mathscr D}$-monoids whose free $\DCat$-monoid on
${\Sigma}$ is $\widehat{V\Sigma}$ for every finite set ${\Sigma}$.
\end{XProposition}

\begin{proof}
(1) We first verify that $V^@$ is a variety of $\DCat$-monoids.

(a)\enspace Closure under products: let $\pi_i: \prod_{i\in I} D_i \to D_i$ be a product of monoids  $D_i\in V^@$. Given a monoid morphism $h:\Psi{\Sigma}^*\to\prod_{i\in
I}D_i$ every morphism $\pi_i\cdot h: \Psi\Sigma^*\ra D_i$ factorizes as $\pi_i\cdot
h=k_i\cdot e_{\Sigma}$ for some $k_i: \widehat{V\Sigma}\ra D_i$. Hence $\langle k_i\rangle: \widehat{V\Sigma} \to \prod D_i$ is the desired morphism with $h=\langle k_i\rangle\cdot e_{\Sigma}$.

(b)\enspace Closure under submonoids: given $m:D\monoto D'$ with
$D'\in V^@$ and a monoid morphism $h:\Psi{\Sigma}^*\to
D$, the homomorphism $m\cdot h$
factorizes through $e_{\Sigma}$ in $\Mon{\DCat}$. Consequently $h$ factorizes
through $e_{\Sigma}$ due to diagonal fill-in:
\[\vcenter\bgroup\hbox\bgroup$\xymatrix{\Psi{\Sigma}^*\ar@{->>}
[r]^{e_{\Sigma}}\ar[d]_h&\widehat{V\Sigma}\ar@{-->}[dl]\ar
[d]^{k}\\D\ar[r]_m&D'}$.\egroup\egroup\] Thus $D\in V^@$.

(c)\enspace Closure under quotients: given $e:D\onto D'$ with $D\in V^@$ and a
homomorphism $h:\Psi{\Sigma}^*\to D'$,
choose a splitting of $e$ in $\Set$, i.e., a function $u:\mathopen|D'\mathclose|\to
\mathopen|D\mathclose|$ with $e\cdot u=\id$. Let $\eta:\Sigma\monoto \under{\Psi\Sigma^*}$ denote the universal map of the free monoid $\Psi\Sigma^*$, and extend the map $u\cdot h\cdot \eta:{\Sigma}\to\mathopen|D\mathclose|$ to a homomorphism $k:\Psi{\Sigma}^*\to D$, which then factorizes as $k=k'\cdot e_{\Sigma}$ because $D$ is in $V^@$. Then the
${\mathscr D}$-monoid morphism $e\cdot k':\widehat{V{\Sigma}}\to
D'$ is the desired factorization of $h$. 
\[
\xymatrix@+1pc{
\Sigma \ar@{ >->}[r]^\eta & \Psi\Sigma^* \ar[dr]|(.34)*+{\labelstyle h} \ar[d]_k \ar@{->>}[r]^{e_\Sigma} & \widehat{V\Sigma} \ar[dl]|(.33)*+{\labelstyle k'} \\
& D \ar@<0.5ex>@{->>}[r]^e & D' \ar@{ >->}@<0.5ex>[l]^u
}
\]
Indeed, using freeness of the $\DCat$-monoid $\Psi\Sigma^*$ and since $u$ is injective it suffices to prove that $u\cdot h\cdot \eta = u \cdot e \cdot k' \cdot e_\Sigma\cdot \eta$ in $\Set$, which holds because
\begin{align*}
u\cdot h\cdot \eta &= u\cdot e\cdot u \cdot h\cdot \eta & \text{($e\cdot u = \id$)}\\
&= u\cdot e \cdot k\cdot \eta & \text{(def. $k$)}\\
&= u\cdot e\cdot k'\cdot e_\Sigma\cdot \eta & \text{(def. $k'$).}
\end{align*}

(2) The ${\mathscr D}$-monoid $\widehat{V\Delta}$ lies in $V^@$ for all finite $\Delta$. Indeed, given a $\DCat$-monoid morphism $h: \Psi\Sigma^* \ra \widehat{V\Delta}$ where $\Sigma$ is finite, we can choose a splitting in $\Set$ (a function $u: \under{\widehat{V\Delta}} \ra \under{\Psi\Delta^*}$ with $e_\Delta\cdot u = \id$), and extend $u\cdot h\cdot \eta$ to a
$\DCat$-monoid morphism $f:\Psi{\Sigma}^*\to\Psi{\Delta}^*
$. By Lemma \ref{lem:commdiag} there is a $\DCat$-monoid morphism $g: \widehat{V\Sigma}\to \widehat{V\Delta}$ such that $g\cdot e_\Sigma = e_\Delta\cdot f$. 
\[
\xymatrix{
\Sigma\ar@{>->}[r]^\eta & \ar@{->>}[d]_{e_\Sigma} \ar[r]^f \ar[dr]^h \Psi\Sigma^* & \Psi\Delta^* \ar@<0.5ex>@{->>}[d]^{e_\Delta}\\
& \widehat{V\Sigma} \ar[r]_g & \widehat{V\Delta} \ar@{ >->}@<0.5ex>[u]^{u}
}
\]
We claim that $g$ is the desired factorization, i.e., $g\cdot e_\Sigma = h$. Using freeness of the $\DCat$-monoid $\Psi\Sigma^*$ and since $u$ is injective it suffices to prove $u\cdot g\cdot e_\Sigma\cdot \eta = u\cdot h\cdot \eta$ in $\Set$, and indeed we have
\begin{align*}
u\cdot g\cdot e_\Sigma \cdot \eta &= u\cdot e_\Delta \cdot f \cdot \eta & \text{(def. $g$)}\\
&= u\cdot e_\Delta \cdot u \cdot h \cdot \eta & \text{(def. $f$)}\\
&= u\cdot h\cdot \eta & \text{($e_\Delta\cdot u = \id$).}
\end{align*}
(3) From the definition of $V^@$ and (2) above we immediately conclude that  $\widehat{V\Delta}$ is the free monoid on a finite set $\Delta$ in the variety $V^@$. Hence, since $\widehat{V\Delta}$ is finite, $V^@$ is  a locally finite variety of $\DCat$-monoids.
 \end{proof}

\begin{proof}[Proof of Theorem \ref{T:Eilenberg}]
The above map $V\mapsto V^@$ defines the desired isomorphism. To see this, we describe its inverse $W\mapsto W^\square$. Let $W$ be a locally finite variety of $\DCat$-monoids with free monoids $e_{\Sigma}:\Psi{\Sigma}
^*\epito D_{\Sigma}$ in $W$. Define an object-finite variety
$W^\square$ of languages in $\Cat$ by forming, for each finite $\Sigma$, the dual local variety
$W^\square{\Sigma}\hookrightarrow\rho{T_{\Sigma}}$ of
$e_{\Sigma}$, i.e.,
$\widehat{W^\square\Sigma} \cong D_\Sigma$.
To verify that $W^\square$ is a
subfunctor of $\rho T$, consider a function
$f_0:{\Delta}\to{\Sigma}$ in $\Set_f$ and the corresponding $\DCat$-monoid morphism $f= \Psi f_0^*: \Psi\Delta^*\ra\Psi\Sigma^*$. Since $D_\Delta$ is the free monoid on $\Delta$ in $W$, we get a unique $\DCat$-monoid morphism $g: V_\Delta\ra V_\Sigma$ with 
with $g\cdot e_{\Delta}=e_{\Sigma}\cdot f$. By Lemma \ref{lem:commdiag} we dually get a $T_\Delta$-coalgebra homomorphism $(W^\square\Sigma)^f \ra W^\square \Delta$, which implies that $W^\square$ is a subfunctor of $\rho T$ by Lemma \ref{lem:subfunc}.

From Proposition \ref{P:fin} and the definition of $(\mathord{-})^@$ and $(\mathord{-})^\square$ it is clear that $(V^@)^\square = V$. To show that $(W^\square)^@=W$, observe first that the varieties $(W^\square)^@$ and $W$ have by definition the same finitely generated free $\DCat$-monoids $D_\Sigma$, and hence contain the same finite $\DCat$-monoids. Moreover, both varieties are locally finite and hence form the closure (in the category $\Mon{\DCat}$) of their finite members under filtered colimits. It follows that $(W^\square)^@=W$, as claimed.

We conclude that $V\mapsto V^@$ defines a bijection between the lattices of object-finite varieties of languages (ordered by objectwise inclusion) and locally finite varieties of $\DCat$-monoids (ordered by inclusion). Moreover, clearly this bijection preserves and reflects the order, so it is a lattice isomorphism.
\end{proof}

We now turn to the Eilenberg theorem for simple varieties. For the proof we need to extend the right-derivative construction of Notation \ref{not:der}. 

\begin{XNotation}\label{not:qx} Let $(A,\alpha)$ be an $L_\Delta$-algebra and $e_A: \Psi\Delta^*\ra A$ the initial homomorphism, see Remark \ref{rem:assalg}. For $x\in\under{\Psi\Delta^*}$ let  $(A,\alpha)_x$ be the $L_\Delta$-algebra with the same states $A$, the same transitions $\alpha_a$, but initial state $e_A(x)$. For a finite $T_\Delta$-coalgebra $Q$ we define the $T_\Delta$-coalgebra $Q_x$ by
\[ \widehat{Q_x} = \widehat{Q}_{\rev_\Delta(x)}.\] 
\end{XNotation}

\begin{XLemma}\label{lem:frcom} Let $x\in \under{\Psi\Delta^*}$.

(a) Every homomorphism $h: Q\ra Q'$ between finite $T_\Delta$-coalgebras is also a homomorphism $h: Q_x\ra Q'_x$.

(b) If $Q\monoto \rho T_\Delta$ is a finite local variety, then a $T_\Delta$-coalgebra homomorphism from $Q_x$ to $Q$ exists for every $x\in\Psi\Delta^*$.

(c) For every finite $T_\Sigma$-coalgebra $Q$ and $\DCat$-monoid morphism $f:\Psi\Delta^*\ra \Psi\Sigma^*$ we have
\[ (Q^f)_x = (Q_{fx})^f.\]
\end{XLemma}

\begin{proof}
For (a) notice first that every homomorphism $h: (A, \alpha) \to (A',\alpha')$ of $L_\Delta$-algebras yields a homomorphism $h: (A,\alpha)_x \to (A', \alpha')_x$ since by initiality of $\Psi\Delta^*$ we have $h \cdot e_A(x) = e_{A'}(x)$. Now given a homomorphism $h: Q \to Q'$ between finite $T_\Delta$-coalgebras, we have its dual $L_\Delta$-algebra homomorphism $\wh h: \wh{Q'} \to \wh Q$. Hence, we have the $L_\Delta$-algebra homomorphism $\wh h: \wh{Q_x'} = \wh{Q'}_{\rev_\Delta(x)} \to \wh Q_{\rev_\Delta(x)} = \wh{Q_x}$ and by duality the desired $T_\Delta$-coalgebra homomorphism $h: Q_x \to Q_x'$. 

For (b) see \cite[Prop.~4.31]{GET1} and its proof. It remains to prove (c). We first prove that 
\begin{equation}\label{eq:afx}
(A^f)_x = (A_{fx})^f
\end{equation} for all $L_\Sigma$-algebras $A=(A,\alpha)$ and $\DCat$-monoid morphisms $f:\Psi\Delta^*\ra\Psi\Sigma^*$. Indeed, both $(A^f)_x$ and $(A_{fx})^f$ have states $A$ and transitions $\alpha_{fa}$ for $a\in\Delta$. Moreover, the initial state of $(A^f)_x$ is $e_{A^f}(x)$, and the initial state of $(A_{fx})^f$ is $e_A(fx)$. Hence, by Lemma \ref{lem:pre}(c), $(A^f)_x$ and $(A_{fx})^f$ have the same initial state.

Now let $Q$ be a finite $T_\Sigma$-coalgebra. Then
\begin{align*}
\widehat{(Q^f)_x} &= (\widehat{Q^f})_{\rev_\Delta(x)} & \text{(def. $(\mathord{-})_x$)}\\
&= (\widehat{Q}^{f^\dag})_{\rev_\Delta (x)} & \text{(def. $Q^f$)}\\
&= (\widehat{Q}_{f^\dag\rev_\Delta(x)})^{f^\dag} & \text{(by \eqref{eq:afx})}\\
&= (\widehat{Q}_{\rev_\Sigma\cdot f(x)})^{f^\dag} & \text{(def $f^\dag$)}\\
&= (\widehat{Q_{fx}})^{f^\dag} & \text{(def. $Q_{fx}$)}\\
&= \widehat{(Q_{fx})^f} & \text{(def. $(\mathord{-})^f$).}
\end{align*}
Hence $(Q^f)_x=(Q_{fx})^f$, as claimed.
\end{proof}

\begin{XConstruction}\ignorespaces\label{C:simple} 
Let $V$ be a variety of languages in ${\mathscr C}$
and ${\Sigma}$ a finite alphabet. Given a finite local
variety $i:Q\monoto V{\Sigma}$ we define a
subvariety $V'\monoto V$ as follows. To define
$V'{\Delta}$ for all finite $\Delta$, consider, for every ${\mathscr D}$-monoid
homomorphism $f:\Psi{\Delta}^*\to\Psi{\Sigma}^*$, the
$T_\Delta$-coalgebra homomorphism $\ol
Vf:(V{\Sigma})^f\to V{\Delta}$, see
Notation \ref{not:vat}. Then factorize the coalgebra homomorphism $[\ol Vf\cdot
i]:\coprod_{f:\Psi{\Delta}^*\to\Psi{\Sigma}^*}Q^f\to
V{\Delta}$ as in Remark~\ref{R:fact}: 
\begin{equation}\label{eq:fact}[\overline
Vf\cdot i]\,\,\equiv\,\,\vcenter\bgroup\hbox\bgroup$\xymatrix
{\coprod\limits_{\hbox to0pt{\hss$\scriptscriptstyle\mathstrut
f:\Psi{\Delta}^*\to\Psi{\Sigma}^*$\hss
}}^{\scriptscriptstyle\mathstrut}Q^f\ar@{->>}[r]^{e_{\Delta
}}&V'{\Delta}\ar@{ >->}[r]^-{m_{\Delta}}&V{\Delta
}}$.\egroup\egroup\end{equation} \end{XConstruction}

\begin{XLemma}\ignorespaces\label{L:simple} $V'$ is a subfunctor of
$V$ (via the $m_{\Delta}$'s) and forms a simple variety with
$Q\subseteq V'{\Sigma}$. \end{XLemma}

\begin{proof} (1)~For every finite alphabet $\Gamma$ and every ${\mathscr D}$-monoid
morphism $g:\Psi{\Gamma}^*\to\Psi{\Delta}^*$ we
prove that there exists a $T_\Gamma$-coalgebra homomorphism
\begin{equation}
\label{eq:4.2}g':(V'{\Delta})^g\to
V'{\Gamma}\quad\hbox{with}\quad m_{\Gamma}\cdot
g'=\ol Vg\cdot m_{\Delta}.
\end{equation} 
By Lemma \ref{lem:subfunc}, this implies in particular that $V'$ is a subfunctor of $V$. Denote by
$p:\coprod_{f:\Psi{\Delta}^*\to\Psi{\Sigma}^*}Q^{f\cdot g}\to
\coprod_{h:\Psi{\Gamma}^*\to\Psi{\Sigma}^*}Q^h$ the $T_\Gamma$-coalgebra
homomorphism whose $f$-component is the coproduct injection of
$h=f\cdot g$. Note that $\coprod_f Q^{f\cdot g} = (\coprod_f Q^f)^g$ by Lemma \ref{lem:qfprops}.
Hence we have the following diagram of $T_\Gamma$-coalgebra homomorphisms:
\begin{equation}\label{eq:g'}
\vcenter\bgroup\hbox\bgroup$\xymatrix{(\coprod_fQ^f)^g\ar
@{->>}[r]^{e_{\Delta}}\ar[d]_p&(V'{\Delta})^g\ar@{ >->}[r]^-{m_{\Delta}}\ar@{-->}[d]^{g'}&(V{\Delta})^g\ar
[d]^{\ol Vg}\\\coprod_hQ^h\ar@{->>}[r]_{e_{\Gamma
}}&V'{\Gamma}\ar@{ >->}[r]_-{m_{\Gamma}}&V{\Gamma
}}$\egroup\egroup\end{equation} 
The outside of the diagram commutes because the $f$-components of the upper and lower path are $T_\Gamma$-coalgebra homomorphisms from the finite $T_\Gamma$-coalgebra $Q^{f\cdot g}$ to $V\Gamma \monoto \rho T_\Gamma$, and $\rho T_\Gamma$ is the terminal locally finite coalgebra. The desired $T_\Gamma$-coalgebra homomorphism $g'$ is
obtained via diagonal fill-in in $\Coalg{T_\Gamma}$, see Remark~\ref{R:fact}.

(2)~$V'{\Delta}$ is a local variety for
every ${\Delta}$. Indeed, by definition $V'\Delta$ it is a subcoalgebra of $V\Delta$ and hence of $\rho T_\Delta$. To prove closure under right derivatives, use
Proposition~\ref{P:right}: since $V\Delta$ and $Q$ are local varieties, we have $T_\Delta$-coalgebra
homomorphisms $h_a:(V{\Delta})_a\to V{\Delta}$ for all $a\in\Delta$ and
$T_\Sigma$-coalgebra homomorphisms 
$k_x:Q_x\to Q$ for all $x\in\Psi\Sigma^*$, see Lemma \ref{lem:frcom}(b). Moreover,
\[ (\coprod_f Q^f)_a = \coprod_f (Q^f)_a = \coprod_f (Q_{fa})^f \]
by Lemma \ref{lem:frcom}(c) and since the construction $(\mathord{-})_a$ clearly commutes with coproducts. Hence we have the following diagram of $T_\Delta$-coalgebra homomorphisms
\[\vcenter\bgroup\hbox\bgroup$\xymatrix{(\coprod\limits
_{\hbox to0pt{\hss$\scriptscriptstyle\mathstrut
f:\Psi{\Delta}^*\to\Psi{\Sigma}^*$\hss
}}^{\scriptscriptstyle\mathstrut}Q^f)_a\ar@{->>}[r]^{e_{\fam
1\Delta}}\ar[d]_{\coprod_fk_{fa}}&(V'{\Delta})_a\ar@{
>->}[r]^{m_{\Delta}}\ar@{-->}[d]^{h'_a}&(V{\Delta
})_a\ar[d]^{h_a}\\\coprod\limits_{\hbox to0pt{\hss
$\scriptscriptstyle\mathstrut f:\Psi{\Delta}^*\to\Psi{\fam
1\Sigma}^*$\hss}}^{\scriptscriptstyle\mathstrut}Q^f\ar
@{->>}[r]_{e_{\Delta}}&V'{\Delta}\ar@{
>->}[r]_{m_{\Delta}}&V{\Delta}}$.\egroup\egroup\] 
whose outside commutes by a finality argument analogous to (1). Diagonal fill-in yields a $T_\Delta$-coalgebra homomorphism $h_a': (V'\Delta)_a \ra V'\Delta$, which shows that $V'\Delta$ is closed under right derivatives by Proposition \ref{P:right}.

(3) $V'$ is a variety of languages. Indeed,
apply Theorem~\ref{T:pre} to conclude from \eqref{eq:4.2}
that $V'$ is closed under preimages. Moreover $Q\subseteq V'{\Sigma}$, due to the
possibility of choosing $f=\id_{\Psi \Sigma^*}$ in
(\ref{eq:fact}) for the case ${\Delta}={\Sigma}$.

(4) $V'$ is object-finite. Note first that for every
finite alphabet ${\Delta}$ there exist only finitely many
preimages $Q^f$, where
$f:\Psi{\Delta}^*\to\Psi{\Sigma}^*$ ranges over
all ${\mathscr D}$-monoid morphisms: indeed, $Q$ is finite and the coalgebra $Q^f$ has the same set of states $Q$. Choose $f_1$, \dots,~$f_n$
such that each $Q^f$ is equal to $Q^{f_i}$ for some $i$. Then consider the $T_\Delta$-coalgebra homomorphism $t:\coprod_{f:\Psi{
\Delta}^*\to\Psi{\Sigma}^*}Q^f\to\coprod_{i=1}^nQ^{f_i}$ whose $f$-component is the coproduct injection of $Q^{f_i}$
whenever $Q^f=Q^{f_i}$. Then $e_{\Delta}=u\cdot t$ for the
obvious morphism $u:\coprod_{i=1}^nQ^{f_i}\to V'{\Delta}$.
Thus, $u$ is surjective since $e_{\Delta}$ is,
proving that $V'{\Delta}$ is finite (being a quotient of
the finite coalgebra $\coprod_{i=1}^nQ^{f_i}$).

(5) $V'$ is simple. Indeed, given a variety
$V''$ of languages with $j:Q\monoto V''\Sigma$ a local
subvariety, we prove $V'\monoto V''$. Denote by
$u'_{\Delta}:V'{\Delta}\monoto\rho{T_{
\Delta}}$ and $u''_{\Delta}:V''\Delta\monoto\rho{T_{
\Delta}}$ the embeddings and by $\overline{V''}f:(V''{
\Sigma})^f\to V''{\Delta}$ the $T_\Delta$-coalgebra
homomorphism of Theorem~\ref{T:pre}. Then the square
\[\vcenter\bgroup\hbox\bgroup$\xymatrix{\coprod\limits_{\hbox
to0pt{\hss$\scriptscriptstyle\mathstrut f:\Psi{\Delta}^*\to
\Psi{\Sigma}^*$\hss}}^{\scriptscriptstyle\mathstrut}Q^f\ar
@{->>}[r]^{e_{\Delta}}\ar[d]_{[\overline{V''}f\cdot
j]}&V'{\Delta}\ar@{ >->}[d]^{u'_\Delta}\ar@{-->}[dl]\\V''{\fam
1\Delta}\ar@{ >->}[r]_{u_\Delta''}&\rho{T_{\Delta}}}$\egroup
\egroup\] commutes due to $\rho{T_{\Delta}}$ being the
terminal locally finite $T_\Delta$-coalgebra. Diagonal
fill-in yields the desired embedding $V'{\Delta}
\monoto V''{\Delta}$. \end{proof}

\begin{XRemark}
For
every locally finite variety $W$ of ${\mathscr D}$-monoids the
set $W_{ f}$ of all finite members of $W$ is clearly a
pseudovariety of $\DCat$-monoids.  Conversely, for every pseudovariety $W$, we denote by $\langle W\rangle$ the variety generated by $W$, i.e., the closure of $W$ under (infinite) products, submonoids and quotient monoids.
\end{XRemark}

\begin{XLemma}\label{lem:simptolocfin}
For every simple pseudovariety $W$ of $\DCat$-monoids, the variety $\langle W \rangle$ is locally finite and $\langle W\rangle_f = W$. In particular, the simple pseudovarieties of $\DCat$-monoids form a (full) subposet of all locally finite varieties of $\DCat$-monoids via the order-embedding $W\mapsto \langle W \rangle$.
\end{XLemma}

\begin{proof}
Suppose that the pseudovariety $W$ is generated by the finite $\DCat$-monoid $D$. Then it is easy to see that also the variety $\langle W \rangle$ is generated by $D$, i.e., $\langle W \rangle = \langle D \rangle$. Fix a finite set $\Sigma$ and consider all functions $u: \Sigma\ra \under{D}$. They define a function $\langle u \rangle: \Sigma \ra \under{D}^{\under{D}^\Sigma}$ that extends uniquely to a $\DCat$-monoid morphism $g: \Psi\Sigma^*\ra D^{\under{D}^\Sigma}$. Letting $g= m\cdot e$ be is its factorization in $\Mon{\DCat}$, we get the commutative diagram below:
\[
\xymatrix{
\Sigma \ar[d]_u \ar[dr]^{\langle u\rangle} \ar@{ >->}[r]^\eta & \Psi\Sigma^* \ar[d]^g \ar@{->>}[r]^e & F\Sigma \ar@{>->}[dl]^m \\
D & D^{\under{D}^\Sigma} \ar[l]^{\pi_u} &
}
\]
This shows that $F\Sigma$ (with universal map $e\cdot \eta$) is the free $\Sigma$-generated $\DCat$-monoid in $\langle W \rangle$: it has the universal property w.r.t. $D$ by the above diagram, and hence it has the universal property w.r.t. all monoids in $\langle W \rangle = \langle D\rangle$.
Moreover $F\Sigma$ lies in $W$, being a submonoid of a finite power of $D$. This implies that every finite monoid in $\langle W \rangle$ lies in $W$, since it is a quotient of a finitely generated free monoid.

We conclude that $\langle W \rangle$ is locally finite and $\langle W\rangle_f = W$, and this equation implies immediately that $W\mapsto \langle W\rangle$ is an injective order-embedding
\end{proof}

\begin{proof}[Proof of Theorem \ref{T:EilenbergSim}]
Recall the isomorphism $V\mapsto V^@$ between object-finite varieties
of languages in $\Cat$ and locally finite varieties of ${\mathscr
D}$-monoids from the proof of Theorem~\ref{T:Eilenberg}. In view of Lemma \ref{lem:simptolocfin} it suffices to show that this isomorphism restricts to one between simple varieties of languages and simple
pseudovarieties of ${\mathscr D}$-monoids, that is, 
\[\text{$V$ is simple iff $(V^@)_f$ is simple.}\]

($\Rightarrow$) If $V$ is a simple variety of languages, generated
by ${\Sigma}$, we prove that the pseudovariety $(V^@)_f$ is  generated by the $\DCat$-monoid $\widehat{V\Sigma}$. First apply Construction~\ref{C:simple} to the finite local variety $Q = V\Sigma$. Then for the resulting variety $V'$ we have $V' = V$. Indeed, $V'\seq V$ follows from Lemma~\ref{L:simple}, and  $V\seq V'$ holds because $V$ is generated by $\Sigma$. It follows that for every $\Delta$ the morphism $m_\Delta$ in Construction \ref{C:simple} is an isomorphism, or equivalently, the family of morphisms $\ol V f: (V\Sigma)^f \to V\Delta$, where $f$ ranges over all $\DCat$-monoid morphisms $f: \Psi\Delta^* \to \Psi\Sigma^*$ is collectively strongly epimorphic in $\Cat$.
%
As in the the proof of
Lemma~\ref{L:simple} we choose finitely many homomorphisms
$f_1$, \dots, $f_n:\Psi{\Delta}^*\to\Psi{\Sigma}^*$
with $h=[\ol Vf_i]:\coprod_{i=1}^n(V{\Sigma})^{f_i}\epito
V{\Delta}$ a strong epimorphism. Dually, by Lemma \ref{lem:commdiag} we get a
$\DCat$-submonoid $\widehat{h}: \widehat{V\Delta} \monoto \prod_{i=1}^n \widehat{V\Sigma}$. We
conclude that every finitely generated free monoid
$\widehat{V\Delta}$ of the pseudovariety $V^@$ is a $\DCat$-submonoid of a
finite power of $\widehat{V\Sigma}$. Consequently $(V^@)_f$ is generated by
$\widehat{V\Sigma}$.

($\Leftarrow$) If $V$ is an object-finite variety of languages
such that $(V^@)_f$ (and hence also $V^@$) is generated by a single finite ${\mathscr
D}$-monoid $D$, we prove that $V$ is simple. Put
${\Sigma}=\mathopen|D\mathclose|$, then since $D$ is a
quotient of the free $\Sigma$-generated monoid $\widehat{V\Sigma}$ in $V^@$, it follows that the pseudovariety
$(V^@)_f$ is also generated by $\widehat{V\Sigma}$. Thus, every ${\mathscr D}$-monoid
in $(V^@)_f$ is a quotient of a submonoid of a finite power
$\widehat{V\Sigma}^n$. Consequently, every free algebra
$\widehat{V\Delta}$ of $V^@$ is a submonoid of a finite power
$\widehat{V\Sigma}^n$. (Indeed, given a quotient $e: D'\epito \widehat{V\Delta}$ and a submonoid
$i: D'\monoto \widehat{V\Sigma}^n$, choose a splitting $u: \under{\widehat{V\Delta}} \ra \under{D'}$, $e\cdot u = \id$, in $\Set$.
Since $\widehat{V\Delta}$ is free, we get a ${\mathscr D}$-monoid morphism
$h: \widehat{V\Delta}\ra D'$ which on the generators coincides with $u$. Then
$e\cdot u = \id$ implies $e\cdot h={\id}$, hence
$i\cdot h: \widehat{V\Delta}\ra \widehat{V\Sigma}^n$ is a submonoid.) Consequently, by composition with the projections $\widehat{V\Sigma}^n \ra V\Sigma$ we obtain a
collectively monic collection $g_1,\ldots, g_n: \widehat{V\Delta}\ra \widehat{V\Sigma}$ of ${\mathscr D}$-monoid morphisms. Choose ${\mathscr
D}$-monoid morphisms $f_1$, \dots,
$f_n:\Psi{\Delta}^*\to\Psi{\Sigma}^*$ with
$e_{\Sigma}\cdot f_i=g_i\cdot e_{\Delta}$ (by starting
with a splitting $e_{\Sigma}\cdot v=\id$ in
$\Set$ and extending $v\cdot g_i\cdot
e_{\Delta}\cdot\eta:{\Delta}\to\mathopen|\Psi{
\Sigma}^*\mathclose|$ to a ${\mathscr D}$-monoid morphism).
By Lemma \ref{lem:commdiag} we get a collection of $T_\Delta$-coalgebra homomorphisms
$h_i: (V{\Sigma})^{f_i^\dag}\to V{\Delta}$ that is collectively
strongly epic. Hence the corresponding homomorphism
$[h_i]:\coprod_{i=1}^n(V{\Sigma})^{f_i^\dag}\epito V{\Delta}$
is a strong epimorphism in $\Cat$. We conclude that $V$ is a simple variety generated by
${\Sigma}$: if $V'$ is any variety such
$j: V{\Sigma}\monoto V'{\Sigma}$ is a local subvariety, then $j$ is a $T_\Delta$-coalgebra homomorphism
$j: (V{\Sigma})^{f_i^\dag}\monoto (V'{\Sigma})^{f_i^\dag}$ for
every $i$ by Lemma \ref{lem:qfprops}. Thus we have the diagram of $T_\Delta$-coalgebra homomorphisms
\[
\xymatrix{
\coprod_i (V\Sigma)^{f_i^\dag} \ar[d]_{\coprod j} \ar[r]^<<<<<{[h_i]} & V\Delta \ar@{>->}[dr] \ar@{-->}[d] & \\
\coprod_i (V'\Sigma)^{f_i^\dag} \ar[r] & V'\Delta \ar@{>->}[r] & \rho T_\Delta
}
\]
where the morphism $\coprod_i (V'\Sigma)^{f_i^\dag} \ra V'\Delta$ exists by closure of $V'$ under preimages. Diagonal fill-in shows that $V\Delta\monoto V'\Delta$.
\end{proof}

\begin{XLemma}\label{lem:simvargen}
Let $V$ be a simple variety of languages in $\Cat$. If $V$ is generated by an alphabet $\Sigma$, then $V$ is generated by any alphabet $\Delta$ with $\under{\Delta}\geq \under{\Sigma}$.
\end{XLemma}

\begin{proof} Let $V'$ be any variety with $V\Delta\seq V'\Delta$. Since $V$ is generated by $\Sigma$, it suffices to show $V\Sigma\seq V'\Sigma$ -- then $V\seq V'$ follows. Observe first that there exist $\DCat$-monoid morphisms $e: \Psi\Delta^*\epito \Psi\Sigma^*$ and $m: \Psi\Sigma^*\hookto \Psi\Delta^*$ with $e\cdot m = \id$. Indeed, if $\Sigma\neq \emptyset$, choose functions $m_0: \Sigma\hookto \Delta$ and $e_0: \Delta\epito \Sigma$ with $e_0\cdot m_0 = \id$ in $\Set$ and put $e=\Psi e_0^*$ and $m=\Psi m_0^*$. If $\Sigma=\emptyset$, consider the two (unique) monoid morphisms $m': \emptyset^*  \ra \Delta^*$ and $e': \Delta^*\ra \emptyset^*$  (satisfying $e'\cdot m' = \id$), and put $e=\Psi e'$ and $m=\Psi m'$. It is easy to see that $e$ and $m$ are indeed $\DCat$-monoid morphisms.

Now let $L\in V\Sigma$. By closure of $V$ under preimages we have $L\cdot e\in V\Delta \seq V'\Delta$. Since $V'$ is also closed under preimages, we conclude $L = L\cdot e \cdot m\in V'\Sigma$ and thus $V\Sigma\seq V'\Sigma$.
\end{proof}

\begin{proof}[Proof of Theorem \ref{T:Eil}] 
(1) Let $\mathscr L_{\mathscr C}$
denote the poset of all varieties of languages in ${\mathscr C}$
and $\mathscr L_{\mathscr C}^0$ its subposet of all simple ones.
We prove that $\mathscr L_{\mathscr C}$ is the free
\cpo-completion of $\mathscr L_{\mathscr C}^0$. Note first that $\mathscr L_{\mathscr C}$ is a complete lattice (in particular, a \cpo)
because an objectwise intersection of varieties of languages
$V_i$ ($i\in I$) is a variety $V$. Indeed, the functor
${T_{\Sigma}}Q={O_{\mathscr C}}\times Q^{\Sigma}$
clearly preserves (wide) intersections, thus an intersection of
subcoalgebras of $\rho{T_{\Sigma}}$ in ${\mathscr C}$ is
again a subcoalgebra. And since ${\mathscr C}$ is a variety of
algebras, intersections in ${\mathscr C}$ are formed on the
level of $\Set$. Now, from $V{\Sigma}=\bigcap_{i
\in I}V_i{\Sigma}$ it clearly follows that $V{\Sigma}$
is closed under derivatives. And closure under preimages is also
clear: given $L$ in $\mathopen|V{\Sigma}\mathclose|$ and
$f:\Psi{\Delta}^*\to\Psi{\Sigma}^*$ in $\Mon{\DCat}$, we have $L\cdot f$ in
$\mathopen|V_i{\Delta}\mathclose|$ for all $i$, thus
$L\cdot f\in\mathopen|V{\Delta}\mathclose|$.

Observe that also an objectwise directed union of varieties is a
variety. The argument is the same: since ${\mathscr C}$ is a
variety of algebras, directed unions are formed on the level of
${\Set}$.

It remains to verify the conditions (C1) and (C2) of a free \cpo-completion.

(C1) Every simple variety $V$ is compact in $\mathscr
L_{\mathscr C}$. Indeed, suppose that $V$ is generated by
${\Sigma}$, and let $V'=\bigcup_{i\in I}V_i$ be a directed union with $V\seq V'$.
Then $V{\Sigma}$ a local subvariety of $V'{\Sigma}$,
and since $V{\Sigma}$ is finite and $\mathopen|V'{\Sigma}
\mathclose|=\bigcup_{i\in I}\mathopen|V_i{\Sigma}\mathclose
|$ in ${\Set}$, there exists $i$ with $V{\Sigma}$ a
local subvariety of $V_i{\Sigma}$. Therefore $V\subseteq
V_i$ because $V$ is simple.

(C2)\enspace Every variety $V$ of languages is the directed join (i.e.~directed union) of its simple subvarieties. 
 To show that the set of all simple subvarieties of $V$ is indeed directed, suppose that two simple subvarieties $V_0$ and $V_1$ of $V$ are given. By Lemma \ref{lem:simvargen} we can assume that both varieties are generated by the same alphabet $\Sigma$. In \cite[Corollary 4.5]{GET1} we proved that any finite subset of a local variety $Q'\monoto \rho T_\Sigma$ is contained in a finite local subvariety $Q\monoto Q'$. Letting $Q'=V\Sigma$, this implies that $V\Sigma$ has a finite local subvariety $Q$ containing $V_0\Sigma\cup V_1\Sigma$. Now apply
Construction~\ref{C:simple} to $Q$ to get a simple
variety $V'\subseteq V$ containing $V_0$ and $V_1$.

Finally, by  \cite[Corollary 4.5]{GET1} and Construction \ref{C:simple} again, every language in $V$ is contained in a simple subvariety of $V$. Hence $V$ is the desired directed union.

(2) Let $\mathscr L_{\mathscr D}$ denote the
poset of all pseudovarieties of ${\mathscr D}$-monoids and
$\mathscr L_{\mathscr D}^0$ its subposet of all simple
ones. We again prove that $\mathscr L_{\mathscr
D}$ is a free \cpo-completion of $\mathscr L_{\mathscr D}^0$. First, $\mathscr L_{\mathscr D}$ is a complete lattice (in particular, a \cpo)
because an intersection of pseudovarieties is a pseudovariety.
Observe that also a directed union of pseudovarieties is a
pseudovariety. It remains to verify (C1) and (C2).

(C1) Every simple pseudovariety $W$, generated by a finite
${\mathscr D}$-monoid $D$, is compact in $\mathscr L_{\mathscr
D}$. Indeed, if a directed join (i.e.~directed union) of pseudovarieties
$\bigcup_{i\in I}W_i$ contains $W$ then some $W_i$ contains $D$ and hence $W$.

(C2)\enspace Every pseudovariety $W$ of ${\mathscr D}$-monoids is
a directed union of simple ones. Indeed, clearly $W$ is the union of its simple subvarieties, and this union is directed because two simple subvarieties $W_1, W_2\seq W$ (generated by $D_1$ and $D_2$, respectively) are contained in the simple subvariety of $W$ generated by $D_1\times D_2$.

(3) From Theorem \ref{T:EilenbergSim} we know that $\mathscr L_{\mathscr C}^0\cong \mathscr L_{\mathscr D}^0$. Since $\mathscr L_{\mathscr C}$ is the free \cpo-completion of $\mathscr L_{\mathscr C}^0$ by (1), and $\mathscr L_{\mathscr D}$ is the free \cpo-completion of $\mathscr L_{\mathscr D}^0$ by (2), the uniqueness of completions gives $\mathscr L_{\mathscr C}\cong \mathscr L_{\mathscr D}$.
\end{proof}

\section{Predualities}

In this appendix we prove in detail the preduality of $\mathbf{JSL}_{01}$ and $\mathbf{JSL}$ (the basis of Pol\'aks original Eilenberg-type correspondence \cite{P}) and the preduality of $\BR$ and $\PSet$ (the basis of our new Eilenberg-type correspondence).

\begin{XTheorem}
$\JSLtc$ and $\JSLnc$ are predual.
\end{XTheorem}

\begin{proof}
The desired dual equivalence $\widehat{(\mathord{-})}: (\JSLtc)_f^{op}\xra{\simeq} \JSLnc_f$ is defined on objects by
\[ Q=(Q,\vee,0,1) \quad\mapsto\quad \widehat{Q} = (Q\setminus \{1\},\wedge)\] 
and on morphisms $h: Q\ra R$ by
\[ \widehat{h}: \widehat{R}\ra\widehat{Q},\quad \widehat{h}r = \bigvee_{hq\leq r} q, \]
where $q$ ranges over $Q$. Here and in the following, the symbols $\vee$, $\wedge$, $\leq$, $0$ and $1$ are always meant with respect to  the order of $Q$ or $R$. We need to verify that  $\widehat{(\mathord{-})}$ is (a) a well-defined functor, (b) essentially surjective, (c) faithful and (d) full.\\

\noindent (a1) $\widehat{h}$ is well-defined as a function, that is, it maps the set $\widehat{R}=R\setminus\{1\}$ to $\widehat{Q}=Q\setminus\{1\}$. Indeed, we have for all $r\in R\setminus \{1\}$:
\[ h(\widehat h r) = h(\bigvee_{hq\leq r} q) = \bigvee_{hq\leq r } hq \leq r \tag*{($\ast$)}  \]
which implies $\widehat h r \neq 1$ because $h1=1$.\\

\noindent (a2) $\widehat{h}$ is a $\JSLnc$-morphism, that is, $\widehat h(r\wedge r') =  \widehat hr\wedge \widehat hr'$ holds for all $r,r'\in R\setminus\{1\}$. Here ``$\leq$'' follows from the (obvious) monotonicity of $\widehat h$. For the converse we compute
\[ h(\widehat h r \wedge \widehat{h} r') \leq h(\widehat hr) \wedge h(\widehat hr') \leq r \wedge r'.\]
The first inequality uses monotonicity of $h$ and the second one uses $(\ast)$.
Hence $\widehat hr \wedge \widehat{h}r'$ is among the elements $q$ in the join $\widehat{h}(r\wedge r')=\bigvee_{hq\leq r\wedge r'} q$, which implies $\widehat hr \wedge \widehat{h}r'\leq \widehat{h}(r\wedge r')$.\\

\noindent (a3) The assignment $h\mapsto \widehat{h}$ trivially preserves identity morphisms. For preservation of composition we consider $h: Q\ra R$ and $k: R\ra S$ in $(\JSLtc)_f$ and show $\widehat{kh}(s) = \widehat{h}\cdot \widehat{k}(s)$ for all $s\in S\setminus\{1\}$, i.e.,
\[\bigvee_{kh(q)\leq s} q = \bigvee_{hq \leq \widehat{k}s} q.\]
This equation holds because, for all $q\in Q$, the inequalities $kh(q)\leq s$ and $hq\leq \widehat k s$ are equivalent. Indeed, if $kh(q)\leq s$ then $hq \leq \bigvee_{kr\leq s} r = \widehat{k} s$. Conversely, if $hq\leq \widehat{k}s$ then \[k(hq) \leq k(\widehat k s) = k(\bigvee_{kr\leq s} r) = \bigvee_{kr\leq s} kr \leq s,\] using that $k$ is monotone and preserves joins.\\

\noindent (b) On the level of posets the construction $Q\mapsto \widehat Q$ first removes the top element and then reverses the order. Conversely, we can turn any semilattice\footnote{Note that any nonempty finite semilattice necessarily has a top element, namely the (finite) join of all of its elements.} $P$ in $\JSLnc_f$ to a semilattice $\overline P$ in $(\JSLtc)_f$ by first adding a new bottom element and then reversing the order. Up to isomorphism these two constructions are clearly mutually inverse, so $\widehat{\overline{P}} \cong P$ for all $P$. This proves that $\widehat{(\mathord{-})}$ is essentially surjective.\\

\noindent (c) Given $h: Q\ra R$ in $(\JSLtc)_f$  we claim that, for all $q\in Q$,
\[ hq = \bigwedge_{q\leq \widehat{h}r} r \tag*{($\ast\ast$)}\]
where $r$ ranges over $R\setminus\{1\}$. This immediately implies that  $\widehat{(\mathord{-})}$ is faithful. To show ``$\leq$'' let $r\in R\setminus\{1\}$ with $q\leq \widehat{h}r$. Since $h$ is monotone and preserves joins, we have 
\[hq\leq h(\widehat h r) = h(\bigvee_{hq'\leq r} q') = \bigvee_{hq'\leq r} hq' \leq r.\]
For the converse note that
\[q \leq \bigvee_{hq'\leq hq} q' = \widehat{h}(hq).\]
Hence $hq$ is one of the elements $r$ occuring in the meet ($\ast\ast$), which means that this meet is $\leq hq$.\\

\noindent (d) Given $g: \widehat{R}\ra\widehat{Q}$ in $\JSLnc_f$ we need to find $h: Q\ra R$ in $(\JSLtc)_f$ with $g=\widehat h$. First extend $g: R\setminus\{1\}\ra Q\setminus\{1\}$ to a map $\ol g: R\ra Q$ by setting $\ol g1 = 1$. Then $\ol g$ preserves all meets of $R$ because $g$ preserves all non-empty meets of $R\setminus\{1\}$. Inspired by (c) we define
\[ hq = \bigwedge_{q\leq \ol gr} r, \]
where $r$ ranges over $R$. Let us show that $h$ indeed defines a $\JSLtc$-morphism. First, $h$ preserves $0$ and $1$ because
\[ h0 = \bigwedge_{0\leq \ol gr} r = \bigwedge_r r = 0, \]
and 
\[ h1 = \bigwedge_{1\leq \ol gr} r = 1\]
In the last equation we use that $\ol g 1 = 1 $, and no other element of $R$ is mapped to $1$ by $\ol g$ since the codomain of $g$ is $Q\setminus\{1\}$.
For preservation of joins, $h(q\vee q') = hq \vee hq'$ for all $q,q'\in Q$, first note that ``$\geq$'' follows from the (obvious) monotonicity of $h$. For the other direction we compute
\[q \leq \bigwedge_{q\leq \ol gr} \ol gr = \ol g(\bigwedge_{q\leq \ol g r}  r) = \ol g(hq)\]
and analogously $q'\leq \ol g(hq')$. Hence
\[ q\vee q' \leq \ol g(hq)\vee \ol g(hq') \leq \ol g(hq\vee hq').\]
The last step uses the monotonicity of $\ol g$. So $hq \vee hq'$ is among the elements $r$ in the meet defining  $h(q\vee q') = \bigwedge_{q\vee q'\leq \ol g r} r$, which implies $h(q\vee q')\leq hq\vee hq'$.

Finally, we prove $g=\widehat h$, i.e.,
\[ gr = \bigvee_{hq\leq r} q \]
for all $r\in R\setminus \{1\}$. To show ``$\geq$'' take any $q\in Q$ satisfying $hq  \leq r$. Then
\[ q\leq \bigwedge_{q\leq\ol g r'} \ol gr' = \ol g(\bigwedge_{q\leq\ol g r'} r') = \ol g(hq) \leq \ol g r = gr. \]
For ``$\leq$'' note first that
\[ h(gr) = \bigwedge_{gr\leq \ol gr'} r' \leq r.\]
The last step uses that $r$ is one of the elements $r'$ over which the meet is taken. Hence $gr$ is one of the elements $q$ in the join $\bigvee_{hq\leq r} q$, so $gr \leq \bigvee_{hq\leq r} q$.
\end{proof}

\begin{XRemark}(a) Every non-unital boolean ring is a distributive lattice with $0$ where the meet is multiplication and the join is $x\vee y = x + y + x\cdot y$. Hence every \emph{finite} non-unital boolean ring has a unit $1$, the join of all its elements. However, homomorphisms between finite non-unital boolean rings need not preserve the unit.

(b) The category $\mathbf{UBR}$ of \emph{unital} boolean rings and unit-preserving ring homomorphisms is isomorphic to the category $\BA$ of boolean algebras (and hence predual to $\Set$). Recall that under this isomorphism $+$ corresponds to exclusive disjunction and $\cdot$ to conjuction.
\end{XRemark}

\begin{XTheorem}
$\BR$ and $\PSet$ are predual.
\end{XTheorem}

\begin{proof}
Recall that $(\PSet)_f$ is equivalent to the Kleisli category of the monad $Z \mapsto Z+1$ on $\Set_f$. The dual comonad on $\mathbf{UBR}_f \simeq \Set_f^{op}$ is $MX = X\times \two$ (where $\two=\{0,1\}$ is the two-element boolean ring) with counit 
\[X\times \two \xra{\epsilon_X} X,\quad \epsilon_X(x,b)=x, \]
and comultiplication 
\[X\times \two \xra{\delta_X}  X \times \two \times \two,\quad \delta_X(x,b) = (x, b,b),\]
and it suffices to show that the Co-Kleisli category $\Kl(M)$ of this comonad is isomorphic to $\BR_f$. The desired isomorphism $I: \Kl(M) \xra{\cong} \BR_f$ is identity on objects and takes a Kleisli morphism $f: X\times \two \ra Y$ to the $\BR_f$-morphism
\[ If: X\ra Y,\quad If(x) = f(x,0). \]
It remains to verify that $I$ is a well-defined full and faithful functor.

(1) $If$ is clearly a $\BR_f$-morphism since $0+0=0$ and $0\cdot 0=0$.

(2) $I$ preserves identities: the identity morphism of $X\in \Kl(M)$ is $\epsilon_X$, and
\[  I\epsilon_X(x) = \epsilon_X(x,0) = x = \id_X(x). \]

(3) $I$ preserves composition: the composition of Kleisli morphisms $f: X\times \two \ra Y$ and $g: Y\times \two \ra Z$ is $g\bullet f: X\times \two \ra Z$ where
\[g\bullet f(x,b) = g \circ Mf \circ \delta_X (x,b) = g(f(x,b),b).\]
Therefore
\begin{align*}
I(g\bullet f)(x) &= g\bullet f(x,0)\\
&= g(f(x,0),0)\\
&= g(If(x),0)\\
&= Ig(If(x))\\
&= Ig \circ If(x)
\end{align*}

(4) $I$ is faithful: let $f,g: X\times\two \ra Y$ be Kleisli morphisms with $If=Ig$, i.e., $f(x,0)=g(x,0)$ for all $x$. Then
\begin{align*}
f(x,1) & =  f(x,0) + f(1,1) + f(1,0)\\
&= f(x,0) + 1 + f(1,0)\\
&= g(x,0) + g(1,1) + g(1,0)\\
&= g(x,1)
\end{align*}
which implies $f=g$.

(5) $I$ is full: let $f: X\ra Y$ be a $\BR_f$-morphism. We claim that the map
\[ g: X\times \two \ra Y,\quad g(x,b) = f(x) + b + b\cdot f(1), \]
is a $\mathbf{UBR}_f$-morphism with $Ig=f$. First, the equation $Ig = f$ clearly holds:
\begin{align*}
Ig(x) &= g(x,0)\\
&= f(x) + 0 + 0\cdot f(1)\\
&= f(x).
\end{align*}
To show that $g$ is a $\mathbf{UBR}_f$-morphism we compute
\[ g(0,0) = f(0) + 0 + 0\cdot f(1) = 0 +0 +0 = 0 \]
and 
\begin{align*}
g(1,1) &= f(1)+1+1\cdot f(1)\\
 &= f(1)+f(1)+1 =0+1 = 1.
 \end{align*}
Further,
\begin{align*}
& g(x,b) + g(x',b')\\
& = (f(x) + b +b \cdot f(1)) + (f(x') + b' +b'\cdot f(1))\\
& = f(x+x') + b+b'+(b+b')\cdot f(1)\\
& = g(x+x',b+b')
\end{align*}
and
\begin{align*}
&g(x,b) \cdot  g(x',b')\\
&= (fx + b +b \cdot f(1)) \cdot (fx' + b' +b'\cdot f(1))\\
&= fx \cdot fx' + (fx\cdot b' + fx\cdot b'\cdot f(1))\\
&~~~+ (b\cdot fx'+ b\cdot fx'\cdot f(1)) \\
&~~~+(b\cdot b'\cdot f(1)+b\cdot b'\cdot f(1)\cdot f(1))\\
&~~~+ b\cdot b' + b\cdot b'\cdot f(1)\\
&= fx \cdot fx' + (fx\cdot b' + fx\cdot b') + (b\cdot fx'+ b\cdot fx') \\
&~~~+(b\cdot b'\cdot f(1)+b\cdot b'\cdot f(1))+ b\cdot b' + b\cdot b'\cdot f(1)\\
&= f(x\cdot x') + b\cdot b'+b\cdot b'\cdot f(1)\\
&=g(x\cdot x',b\cdot b').
\end{align*}
In the third step we use $fx\cdot f(1) = fx$,  $fx'\cdot f(1)=fx'$ and $f(1)\cdot f(1) = f(1)$, and in the fourth step the equation $u+u=0$.
\end{proof}

\begin{XRemark}
By composing the equivalences
\[ \BR_f \simeq \Kl(M) \simeq (\Kl(Z\mapsto Z+1))^{op} \simeq (\PSet)_f^{op} \]
of the above proof, one obtains the explicit description of the preduality between $\BR$ and $\PSet$ in Example \ref{E:bool}(e).
\end{XRemark}

\end{document}